\documentclass[oneside,11pt]{article}
\usepackage{amsmath}
\usepackage{dsfont}
\usepackage{amssymb}
\usepackage{amsthm}

\newtheorem{lemma}{Lemma}

\usepackage{natbib}

\def \dsR {\text{$\mathds{R}$}}

\def \svec {\text{\boldmath$s$}}

    \def \mY {\text{\boldmath$Y$}}
\def \zvec {\text{\boldmath$z$}}

\def \gammavec        {\text{\boldmath$\gamma$}}

\def \thetavec        {\text{\boldmath$\theta$}}
\def \varthetavec     {\text{\boldmath$\vartheta$}}

\newcommand{\LN}{\mbox{LN}}
\newcommand{\WN}{\mbox{WN}}
\newcommand{\GP}{\mbox{GP}}
\newcommand{\GRF}{\mbox{GRF}}
\newcommand{\GMRF}{\mbox{GMRF}}
\newcommand{\SPDE}{\mbox{SPDE}}
\newcommand{\WGP}{\mbox{WGP}}
\newcommand{\MCMC}{\mbox{MCMC}}
\newcommand{\FEM}{\mbox{FEM}}

\newcommand{\IWLS}{\mbox{IWLS}}
\newcommand{\DIC}{\mbox{DIC}}
\newcommand{\WAIC}{\mbox{WAIC}}

\newcommand{\nLS}{\mbox{nLS}}
\newcommand{\ES}{\mbox{ES}}
\newcommand{\cylCRPS}{\mbox{$\text{CRPS}_{cyl}$}}
\newcommand{\PWC}{\mbox{PWC}}

\usepackage{bm} 

\newcommand{\bs}{\bm{s}}
\newcommand{\bss}{\bm{s}}
\newcommand{\bz}{\bm{z}}

\usepackage{booktabs}
\usepackage{multirow}

\usepackage[dvipsnames]{xcolor}
\usepackage[figuresright]{rotating}

\usepackage[top=2.5cm,bottom=3cm, right=2.5cm,left=2.5cm]{geometry}
\title{Copula-based models for spatially dependent cylindrical data}

\author{Francesca Labanca$^{1}$, 
	Anna Gottard$^{1}$, and Nadja Klein$^{2}$ \\\\
\small$^{1}$Department of Statistics, Computer Science, Applications, University of Florence, Florence, Italy \\
	\small$^{2}$Scientific Computing Center, Karlsruhe Institute of Technology, Karlsruhe, Germany}

\date{ }
\begin{document}

\maketitle

\begin{abstract}
\noindent Cylindrical data frequently arise across various scientific disciplines, including meteorology (e.g., wind direction and speed), oceanography (e.g., marine current direction and speed or wave heights), ecology (e.g., telemetry), and medicine (e.g., seasonality and intensity in disease onset). Such data often occur as spatially correlated series of intensities and angles, thereby representing dependent bivariate response vectors of linear and circular components. To accommodate both the circular-linear dependence and spatial autocorrelation, while remaining flexible in marginal specifications, copula-based models for cylindrical data have been developed in the literature. However, existing approaches typically treat the copula parameters as constants unrelated to covariates, and regression specifications for marginal distributions are frequently restricted to linear predictors, thereby ignoring spatial correlation. In this work, we propose a structured additive conditional copula regression model for cylindrical data. The circular component is modeled using a wrapped Gaussian process, and the linear component follows a distributional regression model. Both components allow for the inclusion of linear covariate effects. Furthermore, by leveraging the empirical equivalence between Gaussian random fields (GRFs) and Gaussian Markov random fields, our approach avoids the computational burden typically associated with GRFs, while simultaneously allowing for non-stationarity in the covariance structure. Posterior estimation is performed via Markov chain Monte Carlo simulation. We evaluate the proposed model in a simulation study and subsequently in an analysis of wind directions and speed in Germany. 
\end{abstract}

Non-stationary spatial process; Wrapped Gaussian process. \vspace{0.5cm}\\
 \noindent
{\footnotesize{{\bf Corresponding author}: Prof.~Dr.~Nadja Klein, Scientific Computing Center, Karlsruhe Institute of Technology, Zirkel 2, 76131 Karlsruhe, Germany, nadja.klein@kit.edu.\\
\noindent \textbf{Acknowledgments:} The work of Nadja Klein was supported by the German Research Foundation (Deutsche Forschungsgemeinschaft, DFG) through the Emmy Noether grant KL3037/1-1 and the TRR391 within the project A07, grant number 520388526. The first and second authors were partially supported by the MUR-PRIN grant 2022 SMNNKY, CUP B53D23009470006,  by Next Generation EU, Mission 4, Component 2, and the MUR Dept. of Excellence project 2023-2027 ReDS 'Rethinking Data Science' - University of Florence, the European Union - NextGenerationEU - National Recovery.}}

\clearpage
\section{Introduction}
Cylindrical data arise when observations involve pairs of a \emph{circular} (an angle with domain in $[0,2\pi)$, non-Euclidean), and a \emph{linear} variable (such as intensities with domain in a subset of $\dsR$, Euclidean). They appear across several scientific disciplines, including meteorology, ecology, medicine, and biometrics. Common examples include wind direction and speed from wind profilers \citep{Car2009}, wave direction and height from deterministic wave models, and marine current direction and speed recorded by high-frequency radar networks \citep{Mei2021, Lag2025}. Other applications include animal movement telemetry and the study of seasonality in disease onset \citep{Mas2022, Hod2022}. Further applications can be found in \cite{Pew2021}. 
In many environmental and ecological studies, cylindrical data form spatially correlated series, bivariate vectors of angles and intensities observed at different sites \citep[e.g.,][]{Lag2018}. A key challenge lies in jointly modeling these variables by defining probability distributions that capture angular–linear dependence and potential temporal correlation, while remaining flexible yet parsimonious. Existing approaches use copula-based constructions \citep{Joh1978, Lag2019}, circular distributions based on trigonometric sums \citep{Dur2007}, and circulas, the analogue of copulas for directional data \citep{Jon2015}. Recent work introduced circular analogues of Fréchet–Hoeffding bounds for toroidal variables \citep{Oga2023}, but extensions to cylindrical or multivariate cases remain unexplored.

We propose a Bayesian copula-based model for cylindrical data that addresses these challenges. Unlike approaches that define the joint distribution directly on the cylinder, we adopt a wrapped approach in its most general form \citep{Mar2000}, which naturally yields the partially wrapped version on the cylinder $\mathbb{S}^1\times\mathbb{R}$ by wrapping only the first component around the circle.  Our key idea is to recover the latent winding number and model the unwrapped spatial process via copulas, thereby avoiding partial periodicity and enabling flexible dependence modeling.
While covariate-dependent copulas for linear data are well-studied \citep[e.g.,][]{Aca2011, Kle2016a, Vat2018}, their extension to cylindrical settings remains unexplored. Our work fills this gap by incorporating covariate effects into the copula parameter. We focus on parametric copula families to assess tail behaviour: Gaussian (symmetric, tail-independent), Clayton (lower-tail), and Gumbel (upper-tail). Spatial dependence is introduced through latent Gaussian random fields represented via stochastic partial differential equations (SPDEs), allowing efficient inference via sparse precision matrices \citep{Lin2011, Mil2020}.
By exploiting these latent fields 
and using copulas to govern the dependence between the unwrapped circular and the linear variable, we obtain flexible bivariate distributions on $\mathbb{R}^2$. 
We illustrate the methodology on wind direction and speed data from the Deutscher Wetterdienst (DWD), accounting for spatial dependence, circular–linear marginals, and covariate-driven copula structures. Unlike usual Gaussian-based approaches \citep{Lan2019}, our framework accommodates non-Gaussian margins and allows for non-linear dependence between the circular and linear variables. Although our application focuses on wind, the framework generalizes to any spatial circular–linear data.
 
The remainder of the paper is structured as follows. In Sec. \ref{s:model} we introduce the partially wrapped conditional copula regression
models, with details about the main building blocks. Sec. \ref{s:inf} presents Bayesian inference in the proposed partially wrapped conditional copula 
regression models. Sections \ref{s:sim} and \ref{s:appl} empirically study our approach through simulations and its application to wind data.  Sec. \ref{s:discuss} concludes and discusses directions of future research.

\section{Bivariate conditional copula models with mixed circular-linear marginals}
\label{s:model}
 
In this section, we outline our model specification and exemplify it along with specific choices relevant to our application in Sec. \ref{s:appl}. 

Our focus is on bivariate spatial data  $(\varphi_1(\svec),Y_2(\svec))^\top$, where $\varphi_1(\bss)$ and $Y_2(\bss)$ denote spatial stochastic processes observed at locations $\bss \in \mathcal{D}\subset \dsR^2$. Without loss of generality, we assume that $\varphi_1(\svec)$ is a circular response taking values in~$[0, 2\pi)$, and $Y_2(\svec)$ is a linear continuous response taking values in $\mathbb{R}$. Following the partially wrapped approach of \citet{Jon2012}, we define $\varphi_1(\svec)$ as the wrapped counterpart of an underlying unwrapped linear spatial process $Y_1(\svec) \in \mathbb{R}$ for all $\svec$,
\begin{equation}\label{eq:circvar}
    \varphi_1(\svec) \,=\, Y_1(\svec) \;\; \text{mod}\; 2\pi. 
\end{equation}
Knowledge of $\varphi_1(\svec)$ alone is insufficient to recover the linear representation $Y_1(\svec)$. To address this, we introduce a latent \emph{winding number} process $k(\svec)\in \mathbb{Z}$ for all $\svec$, such that $Y_1(\svec)=\varphi_1(\svec) + 2\pi k(\svec)$.
This identity establishes a one-to-one correspondence between $Y_1(\svec)$ and the pair $(\varphi_1(\svec), k(\svec))^\top$. To model the joint behavior of $(\varphi_1(\svec), Y_2(\svec))^\top$, we adopt a data augmentation strategy and instead model $\mY(\svec)=(Y_1(\svec),Y_2(\svec))^\top=(\varphi_1(\svec)+2\pi k(\svec),Y_2(\svec))^\top$.

To build a flexible joint regression model for $\mY(\svec)$ given some covariates $\bz(\svec)\in\dsR^p$, we adopt the conditional copula regression framework of \citet{Kle2016a}, which enables flexible distributional regression by decoupling the specification of marginal distributions of $Y_1(\svec),Y_2(\svec)$ from the modeling of the dependence between the two. The approach relies on Sklar's theorem, which states that we can write the joint conditional distribution
$F_{1,2}(y_{1}(\svec), y_{2}(\svec) \mid \bz(\svec))$  of $\mY(\svec)$ given $\zvec(\svec)$ and $\svec$ as 
$$F_{1,2}(y_{1}(\svec), y_{2}(\svec) \mid \bz(\svec)) = C\!\left( F_1(y_{1}(\svec) \mid \bz(\svec)), \, F_2(y_{2}(\svec) \mid \bz(\svec)) \,\middle|\, \bz(\svec) \right),$$ 
where $F_1(\cdot\mid \zvec)$ and $F_2(\cdot\mid\zvec(\svec))$ are the marginal conditional cumulative distribution functions (CDFs) of $Y_1(\svec),Y_2(\svec)$, and $C(\cdot,\cdot \mid \bz(\svec))$ is the covariate-dependent copula function. Assuming $\mY(\svec)$ to be continuous, this representation is unique and the marginal densities  $p_1(y_{1}(\svec) \mid \bz(\svec))$, $ p_2(y_{2}(\svec) \mid \bz)$, and the copula density $c(\cdot,\cdot \mid \bz(\svec))=\frac{\partial^2}{\partial y_1(\svec)\partial y_2(\svec)}C(\cdot,\cdot \mid \bz(\svec))$ exist. Hence, the joint conditional probability density function (PDF) of $\mY(\svec)$ is 
\begin{align*}\label{eq:cop}
    p_{1,2}(y_{1}(\svec), y_{2}(\svec) \mid \bz(\svec)) &\;=\; c\!\left( F_1(y_{1}(\svec)\mid \bz(\svec)), \, F_2(y_{2}(\svec)\mid \bz(\svec)) \,\middle|\, \bz(\svec) \right) p_1(y_{1}(\svec)\mid \bz(\svec))\,p_2(y_{2}(\svec)\mid \bz(\svec)). 
\end{align*}
The conditional PDF of the partially wrapped (PW) vector $(\varphi_1(\svec),Y_2(\svec))^\top$ is 
\begin{equation}\label{eq:pw}
    p_{1,2}^{PW}(\varphi_1(\svec),y_{2}(\svec)\mid \bz(\svec)) \;=\; \sum_{k(\svec) \in \mathbb{Z}} p_{1,2}(\varphi_1(\svec) + 2\pi k(\svec)\,, \, y_{2}(\svec)\mid \bz(\svec)). 
\end{equation}
see the Supplementary Material (SM) \ref{apx:PWNmarg} for details. Furthermore, the marginal PDF $p_1$ is the joint distribution of $(\varphi_1(\bss), k(\bss))^\top$, and marginalizing with respect to the winding number $k(\bss)$ yields the conditional density of $\varphi_1(\svec)$, as summarized in Lemma \ref{lem:pw}. 
\begin{lemma}\label{lem:pw}
The margins of the $PW$ copula model defined in 
\eqref{eq:pw} are $p_2$ and the wrapped counterpart of $p_1$, i.e., \begin{equation}\label{eq:w}
p_1^W(\varphi_1(\svec)\mid \bz(\svec)) \;=\; \sum_{k(\svec) \in \mathbb{Z}} p_{1}(\varphi_1(\svec) + 2\pi k(\svec) \mid \bz(\svec)).
\end{equation}
\end{lemma}
The proof of Lemma \ref{lem:pw} is given in the SM \ref{app:prooflemma}.

As a key advantage, we can model arbitrary marginal distributions and their dependence separately as functions of the covariates. We focus on parametric marginal distributions $p_1,p_2$ and one-parameter copulas with association parameter $\rho$. We then model each distributional parameter as a function of covariates and space. 
Specifically, we are concerned with choosing the following components: 
(i)  a parametric copula;  
(ii) a parametric conditional PDF for the circular response; 
(iii) a parametric conditional PDF for the linear response;  
(iv) predictor specifications for all distribution parameters; and  
(v) a spatial specification that allows for nonstationarity in both the mean and covariance.

\paragraph{Dependence structure}
To capture different dependence structures, we use the Gaussian, Clayton, and Gumbel copulas (see Table~\ref{t:cop}, SM \ref{a:Table_copula}). The Gaussian copula is radially symmetric and exhibits no tail dependence, implying independence in the limit. Within the Archimedean class, the Clayton copula models lower-tail dependence with coefficient $2^{-1/\rho}$, whereas the Gumbel copula captures upper-tail dependence with coefficient $2-2^{1/\rho}$. Each is fully specified by $\rho$, with different admissible ranges. Rotated Archimedean copulas can model negative dependence, but we omit details here, as they did not yield reasonable fits for our application. To incorporate covariate effects, we let $\rho(\bss) = h_\rho(\eta_{\rho}(\bss))$, where $h_\rho(\cdot)$ is a bijective function ensuring parameter constraints and $\eta_{\rho}(\bss)$ is an unrestricted predictor specified later on. 
For further details on copula families and dependence properties, we refer to \citet{Nel2006}.


\paragraph{Modeling the marginal distribution of $\varphi_1(\svec)$}
Among wrapped distributions, the wrapped normal (\WN) is a common choice because it links the circular response to an underlying linear Gaussian process (\GP), allowing coherent spatial and regression structures \citep{Jon2012,Mar2022}. Consequently, we choose $p_1^W$ to be a $\WN$ distribution induced by $\gamma_1^*(\bm{s})\sim \GP(\mu_{\gamma_1}(\bm{s})\,,\,\mathcal{K}_1(\bm{s},\bm{s}'))$ with mean and variance function $\mu_{\gamma_1}(\bm{s}),\,\mathcal{K}_1(\bm{s},\bm{s}')$, respectively.
We also consider a white noise process $\epsilon(\bm{s})$, modeling measurement error or microscale variation, with independent and identically distributed (i.i.d.) realizations as $N(0, \varsigma_1^2)$. 
Hence, defining $\gamma_1(\bm{s})\sim \GP(0\,,\,\mathcal{K}_1(\bm{s},\bm{s}'))$ a zero-mean \GP, $\varphi_1(\bm{s})= \mu_{\gamma_1}(\bm{s})+ \gamma_1(\bm{s})+\epsilon \,\; (\text{mod} \, 2 \pi)$ follows a Wrapped Gaussian process ($\WGP$) with the same mean and covariance function of $\gamma_1^*(\bm{s})$. 
We define the model for the circular random vector $\bm{\varphi}_1= (\varphi_1(\bm{s}_1),\varphi_1(\bm{s}_2), \ldots ,\varphi_1(\bm{s}_n))^\top$ at some specific locations $\bm{s}_1, \bm{s}_2, \ldots , \bm{s}_n$ in $\mathcal{D}\subset \mathbb{R}^2$, as
\begin{equation}\begin{aligned}\label{eq:wn}
\bm{\varphi}_1 \mid \bm{\mu}_{\gamma_1},
\bm{\gamma}_1,\varsigma_1^2 &\sim \WN\left( \,
\bm{\mu}_{\gamma_1}
+ \bm{\gamma}_1 \,,\, \varsigma_1^2 \mathbf{I}\,\right) \quad \text{ and } \quad 
\bm{\gamma}_1(\bm{s}) &\sim GP\left( 0, \mathcal{K}_1(\bm{s}, \bm{s}')\right), 
\end{aligned}\end{equation}
with the mean vector $\bm{\mu}_{\gamma_1} = (\mu_{\gamma_1}(\bm{s}_1),\mu_{\gamma_1}(\bm{s}_2). \ldots ,\mu_{\gamma_1}(\bm{s}_n))^\top$,  $\bm{\gamma}_1=(\gamma_1(\bm{s}_1),\gamma_1(\bm{s}_2), \ldots,\gamma_1(\bm{s}_n))^\top$, and $\mathbf{I}$ the $n \times n $ identity matrix.
We denote by $h_{\mu_1}(\cdot)$, $h_{\kappa_1}(\cdot)$, and $h_{\tau_1}(\cdot)$, the elementwise mean, spatial-range-related, and  marginal-variance-related parameter response functions. We use the identity for $h_{\mu_1}(\cdot)$ and the exponential the others to ensure positivity of $\kappa_1,\tau_1$. Details on the spatial dependence modeling and  the corresponding predictors, $\eta_{\mu_1}(\bm{s}_i),\eta_{\kappa_1}(\bm{s}_i), \eta_{\tau_1}(\bm{s}_i)
$, are described later in this section.
\paragraph{Modeling the marginal distribution of 
$Y_2(\svec)$}
The distributional choice for the linear margin should be application-specific.
The log-normal distribution provides a reasonable fit to the data. We define $Y_2(\boldsymbol{s})$ as the log-transformed sum of $\gamma_2^*(\bm{s}) \sim \GP(\mu_{2}(\bm{s})\,, \, \mathcal{K}_2(\bm{s}, \bm{s}'))$ and white noise $\epsilon\sim N(0, \varsigma_2^2)$.
Analogously to what proposed for $\varphi_1(\svec)$, we define the model for the linear random vector $\bm{Y}_2=(Y_2(\bm{s}_1),Y_2(\bm{s}_2), \ldots ,Y_2(\bm{s}_n))^\top$ as
\begin{align}\label{eq:ln}
 \bm{Y}_2 \mid 
 \bm{\mu}_{\gamma_2},\bm{\gamma}_2, \varsigma_2^2 &\sim \LN(\,
 \bm{\mu}_{\gamma_2}+ \bm{\gamma}_2 \,,\, \varsigma_2^2 \mathbf{I}\,
)   \quad \text{ and } \quad 
\bm{\gamma}_2(\bm{s}) \sim GP\left( 0, \mathcal{K}_2(\bm{s}, \bm{s}')\right), 
\end{align}
with $\bm{\mu}_{\gamma_2} = (\mu_{\gamma_2}(\bm{s}_1),\mu_{\gamma_2}(\bm{s}_2) \ldots ,\mu_{\gamma_2}(\bm{s}_n))^\top$ and $\bm{\gamma}_2=(\gamma_2(\bm{s}_1),\gamma_2(\bm{s}_2), \ldots,\gamma_2(\bm{s}_n))^\top$. 
We denote $h_{\mu_2}$, $h_{\kappa_2}$, and $h_{\tau_2}$ the elementwise mean,  spatial-range-related, and  marginal-variance-related parameter response functions. We use the identity for the first one, and the exponential for the remaining, to ensure positivity of $\kappa_2,\tau_2$. 
Details about the spatial dependence modeling are given below.
The corresponding predictors are $\eta_{\mu_2}(\bm{s}_i), \eta_{\kappa_2}(\bm{s}_i),
\eta_{\tau_2}(\bm{s}_i)$.

\paragraph{Predictor specifications}\label{ss:pred}
To account for potential covariates and the spatial dependence in the mean and in the covariance functions, we assume linear $\GP s$ with $\mu_{\gamma_1}(\bm{s})=\beta_{\mu_1,0} +\bm{z}_\beta(\bm{s})\bm{\beta}_{\mu_1,1}$ and $\mu_{\gamma_2}(\bm{s})=\beta_{\mu_2,0} +\bm{z}_\beta(\bm{s})\bm{\beta}_{\mu_2,1}$ as mean functions and the following predictor specifications for distributional parameters of the joint PDF,
\begin{equation*}
\begin{aligned}[c]
\eta_{\mu_1}(\bm{s}) &= \beta_{\mu_1,0} + \zvec_{\beta}(\svec)\bm{\beta}_{\mu_1,1}+\gammavec_1(\svec,\zvec_{\kappa}(\svec)),\\
\eta_{\kappa_1}(\bm{s})& = \theta_{\kappa_1,0}+\zvec_{\kappa}(\svec)\bm{\theta}_{\kappa_1,1}, \\
\eta_{\tau_1}(\bm{s})& = \theta_{\tau_1,0},\;\\
\eta_{\rho}(\bm{s}) &= \beta_{\rho,0} + \zvec_{\rho}(\svec)\bm{\beta}_{\rho,1},
\end{aligned}
\phantom{aaaaaa}
\begin{aligned}[c]
\eta_{\mu_2}(\bm{s}) &= \beta_{\mu_2,0} + \zvec_{\beta}(\svec)\bm{\beta}_{\mu_2,1}+\gammavec_2(\svec,\zvec_{\kappa}(\svec)),\\
\eta_{\kappa_2}(\bm{s}) &= \theta_{\kappa_2,0}+\zvec_{\kappa}(\svec)\bm{\theta}_{\kappa_2,1},\\
\eta_{\tau_2}(\bm{s})& = \theta_{\tau_2,0},\\
\phantom{aaaa}
\end{aligned}
\end{equation*}
where $\beta_{\mu_1,0},\beta_{\mu_2,0},\beta_{\rho,0}, \theta_{\tau_1,0},  \theta_{\tau_2,0}$ are intercept terms, $\bm{\beta}_{\mu_1,1},\bm{\beta}_{\mu_2,1},\bm{\beta}_{\rho,1}, \bm{\theta}_{\kappa_1,1},\bm{\theta}_{\kappa_2,1}$ are  regression coefficient vectors for the covariate effects $\zvec_{\beta}(\svec)$ (marginal means),  $\zvec_{\rho}(\bm{s})$  (dependence parameter), and $\zvec_{\kappa}(\svec)$ (spatial range). 
The spatial effects $\gammavec_1,\gammavec_2$ are zero-mean Gaussian Markov random fields (GMRFs) with a non-stationary covariance function, as described later on. 
The proposed framework allows  nonlinear effects both in the margins and the dependence structure. This flexibility parallels that of structured additive predictors \citep[see][for further details]{Woo2017}. Given the limited sample size in our application, we focus on linear covariate effects and the flexible modeling of spatial effects. In this regard, we found it sufficient to model $\kappa_1(\svec),\kappa_2(\svec)$ while keeping  $\tau_1$ and $\tau_2$ constant in our application. 

\paragraph{Modeling spatial effects}
A Gaussian random field (\GRF) is defined as a continuously indexed spatial process with all finite-dimensional vectors multivariate Gaussian.
Following \citet{Lin2011}, expressing a $\GRF s$ 
through an SPDEs allow modeling the spatial effects  as a random field defined in continuous space, while enabling efficient computation via discrete GMRFs. 
Specifically, the stationary Matérn $\GRF$ $\gamma(\bs)$ with $\bs \in \dsR^2$ arises as the solution to the $\SPDE$
\begin{equation*}\label{eq:spde}
 (\kappa^2 - \Delta)^{\alpha_\nu/2} \, (\tau \, \gamma(\bs)) = W(\bs),
\end{equation*}
where $W$ is Gaussian white noise, $\Delta$ is the Laplacian, $\kappa>0$ controls spatial range $\varrho$, that is the distance at which correlation decays to 0.05, $\tau>0$ is related to the marginal variance   $\sigma^2$, and $\nu>0$ controls the smoothness. The stationary solution has Matérn covariance
\begin{equation}\label{eq:matern}
\mathcal{K}(\bm{s}, \bm{s}') = \sigma^2 \frac{1}{2^{\nu-1}\Gamma(\nu)}
\left(
\kappa
\,h(\bm{s},\bm{s}' )\right)^\nu
K_\nu\!\left(\kappa \, h(\bm{s},\bm{s}' )\right),
\end{equation}
with $h(\bm{s},\bm{s}' )=\|\bm{s}-\bm{s}^\prime\|$, and $K_\nu$ the modified Bessel function. The SPDE and Matérn parameters are coupled. Specifically, the Matérn smoothness parameter is $\nu= \alpha_\nu-1$, and the  variance $\sigma^2$ is proportional to $\tau^{-2}\kappa^{-2\nu}$. Nonstationarity arises by letting $\tau(\bss)$ and $\kappa(\bss)$ to vary smoothly over space \citep{Ing2015}. 
Recent work by \citet{Bol2019} extends the $\SPDE$ approach to fractional powers $\alpha_\nu>1$, but only in the stationary case. Nonetheless, a closed-form expression for the $\GMRF$ precision matrix as a function of $\kappa$ and $\tau$ is available only for $\nu = 1$ (i.e., $\alpha_\nu = 2$). For this reason, we set $\nu = 1$ throughout.
For each $\GRF$ $\gamma_\ell(\bs)$ in the two margins, $\ell\in\{1,2\}$, we thus specify the log-linear models
\begin{equation}\label{eq:tau-kappa}
   \log \tau_\ell(\svec) = \theta_{\tau_\ell,0}, 
   \qquad 
   \log \kappa_\ell(\bm{s}) = \theta_{\kappa_\ell,0} + \bm{z}_\kappa(\bm{s})\bm{\theta}_{\kappa_\ell,1},
\end{equation}
where $\bm{z}_\kappa(\bm{s})$ collects covariates driving local variation and $\bm{\theta}_{\kappa_\ell,1}$ are regression coefficients. The parameters $\theta_{\tau_\ell,0}$ and $\theta_{\kappa_\ell,0}$ correspond to the stationary baseline. Provided $\kappa_\ell(\bm{s})$ varies smoothly, the SPDE admits a local Matérn interpretation via nominal approximations

\begin{equation}\label{eq:rho-sigma}
 \varrho_\ell(\bm{s}) \approx \frac{\sqrt{8\nu}}{\kappa_\ell(\bm{s})}, 
 \qquad
 \sigma_\ell^2(\bm{s}) \approx 
 \frac{\Gamma(\nu)}{\Gamma(\alpha_\nu)\,4\pi\,\kappa_\ell(\bm{s})^{2\nu}\tau_\ell^2}.
\end{equation}
A $\GMRF$ approximation is obtained via the finite element method ($\FEM$). The domain is triangulated into $M$ nodes, and the $\GRF$ is approximated by
$   \gamma_\ell(\bm{s}) = \sum_{m=1}^M \psi_m(\bm{s})\gamma_{\mu_\ell, m}$, 
with $\{\psi_m(\bm{s})\}_{m=1}^M$ piecewise linear basis functions with compact support, whose 
accuracy improves with finer meshes, while the Markov property ensures computational tractability.
\section{Bayesian inference}
\label{s:inf}
This section details the prior specification, hyperparameter choices, and posterior evaluation of the proposed hierarchical model. 
Throughout, we denote the matrices of spatial covariates in the mean, 
and copula predictors as
$\bm{Z}_\beta=(\bm{z}_{\beta}(\bm{s}_1),\ldots,\bm{z}_{\beta}(\bm{s}_n))^\top$, 
and $\bm{Z}_\rho=(\bm{z}_{\rho}(\bm{s}_1),\ldots,\bm{z}_{\rho}(\bm{s}_n))^\top$, respectively. 
Following what stated in the previous section, each GRF is approximated through a finite element basis expansion 
with piecewise linear basis functions $\{\psi_m(\bm{s})\}_{m=1}^M$ and Gaussian weights 
$\bm{\gamma}_{\mu_\ell}=(\gamma_{\mu_\ell,1},\ldots,\gamma_{\mu_\ell,M})^\top$ for $\ell=1,2$. 
This yields two GMRFs characterized by the hyperparameter vector 
$\bm{\theta}_\ell=(\theta_{\tau_\ell,0},\,\theta_{\kappa_\ell,0},\,\bm{\theta}_{\kappa_\ell,1}^\top)^\top$.
We collect the basis functions in $\bm{\psi}(\bm{s})=(\psi_1(\bm{s}),\ldots,\psi_M(\bm{s}))^\top$ and define the matrix
$\bm{\psi}=(\bm{\psi}(\bm{s}_1),\ldots,\bm{\psi}(\bm{s}_n))^\top$.
The model parameters are grouped as 
$\bm{\vartheta}_1=(\beta_{\mu_1,0}, \bm{\beta}_{\mu_1,1}^\top,\bm{\gamma}_{\mu_1}^\top, \varsigma_1^2)^\top$, 
$\bm{\vartheta}_2=(\beta_{\mu_2,0}, \bm{\beta}_{\mu_2,1}^\top,\bm{\gamma}_{\mu_2}^\top, \varsigma_2^2)^\top$, 
and $\bm{\vartheta}_\rho=(\beta_{\rho,0}, \bm{\beta}_\rho^\top)^\top$ for the two marginal and the copula components. A graphical illustration of the full hierarchical models is shown in Figure~\ref{f:DAG}.

\subsection{Prior specifications}\label{ss:prior}
We adopt conjugate priors whenever possible. We assign Gaussian priors for the regression coefficients and inverse gamma (IG) priors for their variances. For the overall level of the mean and copula predictors, $\beta_{\mu_\ell,0}$, $\beta_{\rho,0}$, we choose Gaussian priors with zero mean and large variances as a weakly informative prior on the log-scale.
For the remaining coefficients, $\bm{\beta}_{\mu_\ell,1}$ and $\bm{\beta}_{\rho,1}$, we assume independent, homoscedastic priors, with zero-mean Gaussian distributions with variances $\xi_{\mu_\ell}^2$ and $\xi_\rho^2$, respectively. 
The error variance $\varsigma_\ell^2$ and the scaling variances $\xi_{\mu_\ell}^2$ and $\xi_{\rho}^2$ use IG priors with shape and scale equal to 0.001 to obtain a data-driven smoothness.
 Discretization of the spatial domain yields Gaussian 
 vectors $\bm{\gamma}_{\mu_\ell}$ with a sparse precision matrix $\bm{Q}(\bm{\theta}_\ell)$, depending on the $\FEM$ and on the spatial regression coefficients $\bm{\theta}_\ell$, as defined in~\eqref{eq:tau-kappa}.
More details about the construction of $\bm{Q}(\bm{\theta}_\ell)$ can be found in the SM~\ref{apx:GMRF}.

\begin{figure}[ht]
\begin{center}
\centerline{\includegraphics[width=6.5in]{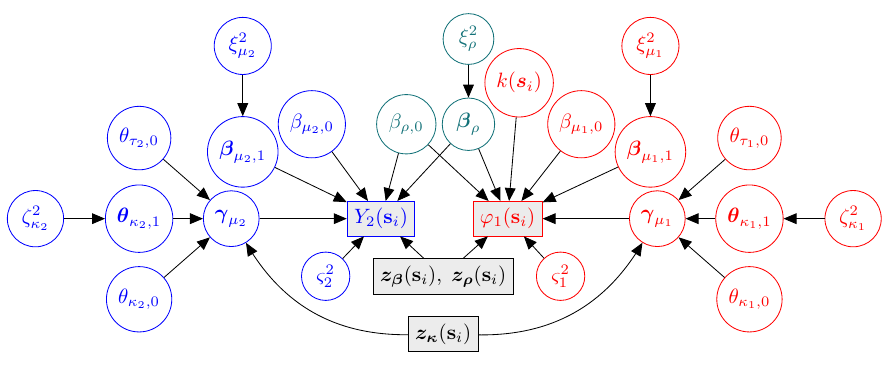}}
\end{center}
\caption{Graphical representation of the full hierarchical model. Shown are the spatially dependent covariates (grey), the model parameters and data $\{\varphi_1(\mathbf{s}_i)\}_{i=1}^n$ related to the wrapped circular random variable red),  the model parameters and data $\{y_2(\mathbf{s}_i)\}_{i=1}^n$ related to the linear random variable (blue), and the copula-parameters (petrol).}\label{f:DAG}
\end{figure}

Here, $\theta_{\tau_\ell,0}$ and $\theta_{\kappa_\ell,0}$ represent the stationary baseline, with \eqref{eq:rho-sigma} giving their local interpretation of range and marginal variance.
We assign to them independent uniform priors, chosen to yield interpretable ranges for the marginal variance and correlation range. To control deviations from stationarity, we assign a standard Gaussian prior to the spatial coefficients $\boldsymbol{\theta}_{\kappa_\ell,1}$, 
with smoothing variance $\zeta_{\kappa_\ell}^2$ endowed with a penalized complexity (PC) prior \citep{Sim2017}. 
The PC prior provides a principled mechanism to penalize model complexity by shrinking towards a predefined 
\emph{base model}, here, a stationary covariance function. 
Hence, non-stationarity is introduced only when supported by the data, ensuring parsimony and interpretability. 
Under the PC prior, a model component is treated as a flexible extension of a simpler base model, 
and the prior density decreases exponentially with the Kullback--Leibler divergence from the base, 
thereby enforcing a constant-rate penalization of complexity. The guiding principles are listed in SM \ref{apx:PC}. \citet{Kle2016b} proved that if the base model is obtainable for $\xi_{\kappa_\ell}^2 \rightarrow 0$, i.e., $\boldsymbol{\theta}_{\kappa_\ell,1}=\mathbf{0}$ in our case, and thus, the field predictor $\eta_{\kappa_\ell}$ is constant, if $\xi_{\kappa_\ell}^2$-prior is constructed according to the PC-prior principles, $\xi_{\kappa_\ell}^2$ follows a Weibull distribution with shape $1/2$ and scale $\lambda$, i.e., $\xi_{\kappa_\ell}^2 \sim \mathrm{Weibull}\!\left(\tfrac{1}{2},\,\lambda\right)$.
The rate parameter $\lambda$ is determined via a user-defined tail probability condition that encodes prior beliefs on the plausible degree of non-stationarity and is elicited as described later.
A prior sensitivity analysis for $\xi_{\mu_1}^2$ and $\zeta_{\mu_1}^2$ in the circular marginal model can be found in \citet{Mar2022}.
 
For the latent winding numbers $k(\bm{s}_i) \in \mathbb{Z}$, assigning a prior over the full integer domain is computationally intractable. 
Following \citet{Jon2020} and \citet{Mar2022}, we adopt a truncated uniform prior with restricted support, i.e., $p(k_i)= \frac{1}{3} \mathbf{1}_{\{k_i \in \{-1,0,1\} \} }$, which provides a practical and accurate approximation for the univariate $\WN$ distribution \citep[][]{Kur2014b}.
Thus, the resulting hierarchical model is the partially wrapped copula ($\PWC$) model with wrapped normal and log-normal margins, specified as
\begin{equation}\label{eq:DAG}
\begin{aligned} 
\big(\varphi_{1}(\bm{s}_i), Y_{2}(\bm{s}_i) \big)^\top &\mid \bm{z_\beta}(\bm{s}_i), \bm{z}_\rho(\bm{s}_i), \vartheta_1,\vartheta_2, \vartheta_\rho  \sim \PWC(\mathbf{\vartheta}_1,\mathbf{\vartheta}_2, \mathbf{\vartheta}_\rho)  \\  
 {\beta}_{\rho,0} & \sim N(0, 100), \;  \bm{\beta}_{\rho,1} \mid  \xi_{\rho}^2 \sim N(\bm{0}, \xi_\rho^2\,\mathbf{I}), \; 
   \xi_{\rho}^2 \sim  IG(0.001, 0.001),\\
  {\beta}_{\mu_\ell,0} & \sim N(0, 10), \; 
  \bm{\beta}_{\mu_\ell,1} \mid  \xi_{\mu_\ell}^2 \sim N(\bm{0}, \xi_{\mu_\ell}^2\,\mathbf{I}), \; 
   \xi_{\mu_\ell}^2 \sim  IG(0.001, 0.001),\\
  \bm{\gamma}_{\mu_\ell} \mid z_\kappa(\bm{s}_i), \bm{\theta}_\ell &\sim N(\mathbf{0}, \bm{Q}^{-1}(\bm{\theta}_\ell)),\\{\theta}_{\tau_1,0}&\sim U(a_{\tau_1}, b_{\tau_1}), \; 
  {\theta}_{\tau_2,0}\sim U(a_{\tau_2}, b_{\tau_2}), 
  \\
  {\theta}_{\kappa_\ell,0}&\sim U( a_{\kappa}, b_{\kappa}), \; 
  \bm{\theta}_{\kappa_\ell,1} \mid \zeta_{\kappa_\ell}^2 \sim N(\bm{0}, \zeta_{\kappa_\ell}^2 \mathbf{I}), \;
  \zeta_{\kappa_\ell}^2 \sim \text{PC}(c
  , \alpha_\zeta
  ),\\
  \varsigma_\ell^2 &\sim IG(0.001, 0.001), \; p(k_i)= \frac{1}{3} \mathbf{1}_{\{k_i \in \{-1,0,1\}  \}}. 
\end{aligned}
\end{equation}
The specific ranges for uniform priors are described next. 

\paragraph{Prior scaling}
To penalize the non-stationarity in the covariance of the margins, we want the uniform priors to ensure that the spatial ranges and marginal variances fall within predefined intervals.
After rescaling the domain, for $\mathcal{D}\subseteq[0,1]^2$, the spatial ranges are in $[0.01,1]$ and marginal variances in $\left[0.01,\,s_{\max}^2\right]$, for a chosen $s_{\max}^2 >0.01$. The latter follows from the nominal approximation in \eqref{eq:tau-kappa}, linking $\varrho(\bm{s}_i)$ to $\kappa(\bm{s}_i)$, and $\sigma(\bm{s}_i)$ to $\tau(\bm{s}_i)$ and $\kappa(\bm{s}_i)$.Hence, $
   \theta_{\tau_\ell,0} \sim \mathcal{U}(-7-\log s_{\max},\,0), 
   \,
   \theta_{\kappa_\ell,0} \sim \mathcal{U}(1,6).
$ 
For circular responses, we choose $s_{max}=2\pi$, since it reflects the maximal support length. For linear response, we refer to the Beaufort Wind Scale ($(0,31.5]$ m/s).
On a logarithmic scale, this gives $s_{max}=3.45$. 
The scale parameter $\lambda$ is calibrated by simulation, using the tail bound
\begin{equation}\label{eq:pcprior}
   \Pr \left( \max_{\bm{s}_i\in\mathcal{D}}\,\big|z_\kappa(\bm{s}_i)\bm{\theta}_{\kappa_\ell}\big|\leq c\right) \geq\ 1-\alpha_\zeta,
\end{equation}
for user–defined constants $\alpha_\zeta\in(0,1)$ and $c>0$, which controls the magnitude of the nonstationary deviation. This ensures that with probability at least $1-\alpha_\zeta$, the nonstationary effect in the covariance
does not exceed the bound $c$.
Under the rescaled domain assumption, based on the range constraint $\rho(\bm{s}_i)\in[0.01,1]$, we choose the bound $c$ as
\begin{equation}\label{eq:c-bound}
   c =\Biggl \lfloor  \tfrac{1}{2}  \Biggm|	
 \log\!\Big(\tfrac{\sqrt{8}}{0.01}\Big)-\log(\sqrt{8}) \Biggm| \Biggr \rfloor.	
\end{equation}

\subsection{Likelihood}\label{ss:lik}
With the specific choices of copula, and marginal distributions, the joint likelihood of our PWC model for a dataset of $n$ observations $\lbrace y_1(\svec_i),\varphi_2(\svec_i),\zvec_\beta(\svec_i),\zvec_\rho(\svec_i),\svec_i\rbrace_{i=1}^n$  is
\begin{equation}\label{eq:lik}
\begin{aligned}
l(\bm{\vartheta}_1,\bm{\vartheta}_2, \bm{\vartheta}_\rho) = \prod_{i=1}^n 
& 
c( \, \Phi(\varphi_1(\bm{s}_i)+2\pi k(\bm{s}_i) \mid z_{\bm{\beta}}(\bm{s}_i), \bm{\vartheta_1}
)
 \,, \, \Phi(\, \log(y_2(\bm{s}_i)) \mid z_{\bm{\beta}}(\bm{s}_i),\bm{\vartheta}_2) \, \mid \, z_{\bm{\rho}}(\bm{s}_i), \bm{\vartheta}_\rho) \cdot
\\
& \qquad \qquad
\phi(\varphi_1(\bm{s}_i)+2\pi k(\bm{s}_i) \mid z_{\bm{\beta}}(\bm{s}_i),\bm{\vartheta_1}) 
\cdot \frac{1}{y_2(\bm{s}_i)}\phi(\, \log(y_2(\bm{s}_i)) \mid z_{\bm{\beta}}(\bm{s}_i), \bm{\vartheta}_2) 
\end{aligned}
\end{equation}
where $\Phi$ and $\phi$ are the univariate Gaussian CDF and pdf, 
with parameters $\bm{\vartheta_1}$ and~$\bm{\vartheta_2}$. 

\subsection{Posterior estimation}
Posterior estimation in our $\PWC$ model becomes tractable by introducing the winding number vector $\bm{k}=(k(\svec_1),\ldots,k(\svec_n))^\top$ as a latent variable. 
Assuming prior independence among model parameters, the log-posterior distribution is
\begin{align*}p(\bm{\vartheta}_1,\bm{\vartheta}_2,\bm{\vartheta}_\rho,&\bm{\theta}_1,\bm{\theta}_2, \xi_{\mu_1}^2,\xi_{\mu_2}^2, \xi_{\mu_\rho}^2, \zeta_{\kappa_1}^2,\zeta_{\kappa_2}^2 \mid \bm{y}, \bm{\varphi}, \bm{k},\bm{Z}_\beta,\bm{Z}_\kappa, \bm{Z}_\rho)
    \; \propto \; l(\bm{\vartheta}_1,\bm{\vartheta}_2, \bm{\vartheta}_\rho)  
       \cdot\\
&p(\bm{\beta}_{\rho} \mid \xi_{\rho}^2)p(\xi_{\rho}^2)\prod_{\ell=1}^2
   p(\varsigma_\ell^2)
    p(\bm{\beta}_{\mu_\ell} \mid \xi_{\mu_\ell}^2)
    p(\xi_{\mu_\ell}^2)
   p(\bm{\gamma}_{\mu_\ell} \mid \bm{\theta}_{\mu_\ell}^2) p(\theta_{\tau_\ell,0})p(\theta_{\kappa_\ell,0})p(\bm{\theta}_{\kappa_\ell,1}\mid \zeta_{\kappa_\ell}^2)
    p(\zeta_{\kappa_\ell}^2), 
\end{align*}
 proportional to the likelihood (Sec. \ref{ss:lik}), and the prior distributions (Sec. \ref{ss:prior}).

We develop a two-stage Bayesian approach based on $\MCMC$ simulations to estimate the model parameters and hyperparameters. This approach fits the marginal models first and, conditional on the fitted marginals, estimates the copula parameter model. This yields substantial computational savings compared with the joint estimation. Previous work \citep[e.g.,][]{Kle2016a} has shown that this two-step procedure does not substantially underestimate uncertainty compared to the joint estimation.
In the first stage, given the latent variables $\bm{k}$, and transforming back the linear one, any standard algorithm for sampling Gaussian process parameters can be used. 
The posterior for $\bm{\beta_{\mu_\ell}}$,$\bm{\gamma}_{\mu_\ell}$ are sampled via Gibbs, since the full conditionals are wrapped normal and normal.
Using the following vector form of the mean predictors for the two margins, 
and the copula predictor, $ \bm{\eta}_{\mu_\ell} = \mathbf{1} \beta_{{\mu_\ell},0} + \bm{Z}_\beta \bm{\beta}_{{\mu_\ell},1} + \bm{\psi} \bm{\gamma}_{\mu_\ell}$, $\bm{\eta}_\rho=\mathbf{1}\beta_{\rho,0}+ \bm{Z}_\rho \bm{\beta}_{\rho,1}$,
we can explicitly write their full conditional distributions as 
\begin{align*}\begin{split}
{\beta_{\mu_\ell,0}}\mid \cdot &
\sim N\left(\left(\frac{n}{\varsigma_\ell^2}+\frac{1}{10}\right)^{-1}\left(\frac{n}{\varsigma_\ell^2}\mathbf{1}^\top( \widetilde{\bm{y}}_\ell -\bm{\eta}_{\mu_\ell}+\bm{1}\beta_{\mu_1,0}
)\right)\,,\, \left(\frac{n}{\varsigma_\ell^2}+\frac{1}{10}\right)^{-1}\right)   \\
{\bm{\beta}_{\mu_\ell,1}}\mid \cdot & 
\sim N\left(\left(\frac{1}{\varsigma_\ell^2} \bm{Z_\beta}^\top\bm{Z_\beta}+ \frac{1}{\xi_{\mu_\ell}^2}\bm{I}\right)^{-1} 
\left( \frac{1}{\varsigma_\ell^2}\bm{Z}_\beta^\top(\widetilde{\bm{y}}_\ell -\bm{\eta}_{\mu_\ell}+\bm{Z}_\beta\beta_{\mu_\ell,1}
)\right) 
\,,\, 
\left(\frac{1}{\varsigma_\ell^2} \bm{Z_\beta}^\top\bm{Z_\beta}+ \frac{1}{\xi_{\mu_\ell}^2}\bm{I}\right)^{-1} \right)  \\
\bm{\gamma}_{\mu_\ell}\mid \cdot & 
\sim N\left( \left(\frac{1}{\varsigma_\ell^2} \bm{\psi}^\top\bm{\psi}+  \bm{Q}(\bm{\theta}_\ell)  \right)^{-1} 
\left( \frac{1}{\varsigma_\ell^2}\bm{\psi}^\top(\widetilde{\bm{y}}_\ell -\bm{\eta}_{\mu_\ell}+\bm{\psi}\gamma_{\mu_\ell,1}
)\right) 
\,,\, 
\left(\frac{1}{\varsigma_\ell^2} \bm{\psi}^\top\bm{\psi}+  \bm{Q}(\bm{\theta}_\ell)  \right)^{-1}  \right) 
\end{split}\end{align*}
\begin{align*}\begin{split}
\varsigma_\ell^2 \mid \cdot &
\sim IG \left(\frac{n}{2}+0.001\,,\, \frac{1}{2} \left(\widetilde{\bm{y}}_\ell-\bm{\eta}_{\mu_\ell}\right)^\top \left(\widetilde{\bm{y}}_\ell-\bm{\eta}_{\mu_\ell}\right) +0.001 \right)\\
\xi_{\mu_\ell}^2 \mid \cdot &
\sim IG \left(\frac{p_{z_\beta}}{2}+0.001\,,\, \frac{1}{2} \bm{\beta}_{\mu_\ell,1}^\top \bm{\beta}_{\mu_\ell,1} +0.001 \right)\\
\xi_{\rho}^2 \mid \cdot &
\sim IG \left(\frac{p_{z_\rho}}{2}+0.001\,,\, \frac{1}{2} \bm{\beta}_{\rho,1}^\top \bm{\beta}_{\rho,1} +0.001 \right)
\end{split}\end{align*}
where $\widetilde{\bm{y}}_1=\bm{\varphi}_1+2\pi\bm{k}$, $\widetilde{\bm{y}}_2=\log(\bm{y_2})$ and $p_{z_\beta}$, $p_{z_\rho}$ are the covariates number in the mean and in the copula, respectively. See SM \ref{apx:fullCon} for calculation details.
The identifiability challenges of covariance parameters in GRFs are well documented \citep{Tan2021}. These issues arise in our setting as well, and since the components of $\thetavec_\ell$ exhibit strong posterior correlation, they are updated jointly in a single Metropolis–Hastings block, using the Robust Adaptive Metropolis algorithm \citep[] []{Vih2012}  with t-Student proposals. This method adaptively estimates the shape of the target distribution while enforcing a desired acceptance rate. We follow the recommendations of $23.4\%$ as the acceptance rate in multidimensional targets and $ 44\%$ in unidimensional targets \citep{Gel1997} for the $\GMRF$ smoothing variances $\zeta_{\kappa_\ell}^2$. To avoid negative invalid proposals, we reparameterize on the log-scale and approximate the log-full conditional $\log p(\log \zeta_{\kappa_\ell}^2\mid \cdot)$, rather than $\log p(\zeta_{\kappa_\ell}^2 \mid \cdot)$.
Each latent winding number $k(\bm{s}_i)$ is updated at each iteration using a separate Metropolis step. 
At the $t$-th iteration, given a current state $k^{[t]}(\bm{s_i})$, the proposal is drawn uniformly from $\big\{ k^{[t]}(\bm{s}_i)-1 \,,k^{[t]}(\bm{s}_i),\,k^{[t]}(\bm{s}_i)+1 \big\}$. Considering the prior choice in Sec.~\ref{ss:prior}, the initial values are set as $k^{(0)}(\bm{s}_i)=0$ for all $\bm{s}_i \in \mathcal{D}$.

In the second stage, we work with the estimated copula data $\bm{u}_1=(u_{1}(\svec_1),\ldots, u_{1}(\svec_n))^\top$, and $\bm{u}_2=(u_{2}(\svec_1), \ldots, u_{2}(\svec_n))^\top$ defined as $u_{\ell}(\svec_i)=F_{\ell}(y_{\ell}(\svec_i)\mid \widehat{\varthetavec}_\ell)$, $i=1,\ldots,n$, $\ell=1,2$,  where $\widehat\varthetavec_\ell$ are the posterior mean estimates of the %
parameters for margin $\ell$. For sampling the posterior
copula regression coefficients, we rely on Metropolis-Hasting steps with Iteratively Weighted Least Squares (\IWLS) proposal densities \citep{Kle2015}. The regression coefficients $\boldsymbol{\beta}_{\rho}^{[t]}$ at the $t$-th iteration are proposed from  $q\left(\bm{\beta}_{\rho}^* \mid \bm{\beta}_{\rho}^{[t]} \right)~=~ \mathcal{N}\left(\mu^{[t]} \; , \; {P^{-1}}^{[t]}  \right)$,~with
$$
P^{[t]}=\bm{Z}_\rho^\top W^{[t]} \bm{Z}_\rho + \frac{1}{\xi_{\rho}^2} \bm{I}, \text{ and } \mu^{[t]}=
({P^{[t]}})^{-1} \bm{Z}_\rho^\top W^{[t]} \left( \bm{Z}_\rho\bm{\beta}_\rho^{[t]} +(W^{[t]})^{-1}\bm{v}^{[t]}  
\right),
$$
where $\bm{W}^{[t]}$ is a working weight matrix having the negative second derivatives of the log-likelihood $\log(l)$ with respect to the predictor $\bm{\eta}_{\rho}$ on the diagonal, i.e., $w_{ii}= - \frac{\partial^2 }{(\partial \eta_i )^2}^{[t]} \log(l)$ and zeros otherwise; and $\bm{v}^{[t]}=\frac{\partial }{\partial \bm{\eta}}^{[t]} \log(l)$ is the score vector. The working weights and score vector are determined by the chosen copula distribution, thereby ensuring automatic adaptation to the form of the full conditional and eliminating the need for manual tuning.

\subsection{Model choice}
For model comparison, we rely on the deviance information criterion ($\DIC$) and the Watanabe–Akaike information criterion ($\WAIC$), both of which approximate out-of-sample predictive accuracy and are asymptotically equivalent to leave-one-out cross-validation (see \citet{Spi2002,Wat2010}. Let $\bm{\vartheta}^{[1]}, \ldots, \bm{\vartheta}^{[T]}$ denote the full parameter vectors sampled via $\MCMC$ from the posterior, the $\DIC$ is given by 
$\DIC=
\frac{2}{T} \sum_t D(\bm{\vartheta}^{[t]})-D(\frac{1}{T} \, \sum_t\bm{\vartheta}^{[t]})$ where $D(\bm{\vartheta}) = -2 \log p(y \mid \bm{\vartheta}).
$ The $\WAIC$ is given by  $\WAIC = -2 \, (\text{lppd} - p_{\text{WAIC}})$, where $\text{lppd} = \sum_{i=1}^n \log \int p(y_i \mid \bm{\vartheta}) \, p(\bm{\vartheta} \mid y) \, d\bm{\vartheta}$ and $ p_{\text{WAIC}} = \sum_{i=1}^n 2 \, \text{var}_{\bm{\vartheta}}\!\big( \log p(y_i \mid \bm{\vartheta}) \big)$ \citep[see][]{Gel2013,Veh2017}. More details are provided in the SM \ref{s:evmet}.

In our application, we use both $\DIC$ and $\WAIC$ to inform model choice, striking a balance between predictive accuracy and parsimony. Thanks to the limited number of covariates, in the application in Sec.~\ref{s:appl}, we explore all possible combinations of candidate predictors.

\section{Simulation study}
\label{s:sim}
We conduct a simulation study to assess the ability of the proposed methodology to identify tail dependence and covariate effects in the copula parameter. Specifically, we compare the performance of three copula families, Gaussian, Clayton, and Gumbel, representing the absence of tail dependence, lower-tail~dependence, and upper-tail dependence, respectively, and both, covariate-dependent and fixed. 

\paragraph{Simulation design}
We consider six scenarios, defined by combinations of copula families (Gaussian, Clayton, or Gumbel) and dependence specifications (constant or covariate-dependent copula parameters). The marginal specifications follow those in Sec.~\ref{s:model}. For each setting, we consider sample sizes \( n = 250, 500, 750 \) and perform \( R = 100 \) replications. Spatial locations are sampled uniformly in \( \mathcal{D} = [0,1]^2 \), using a fixed mesh with \( M = 703 \) nodes. Covariates are defined as \( z_\beta(\svec) = 2\sin(2\pi s_1)\sin(4\pi s_2) \) for marginal means and \( z_\kappa(\svec) = 1/2 + \sin(2\pi s_1)\cos(4\pi s_2) \) for covariance components. In the \emph{constant dependence} case, we set \( \eta_\rho = 0.577 \), yielding a target correlation parameter \(\rho \approx 0.5 \) for the Gaussian copula, \(\rho \approx 1.781 \) for the Clayton copula and \(\rho \approx 2.781 \) for the Gumbel copula. 
In the \emph{covariate-dependent} scenario, we define \( \eta_\rho(\svec) \approx 0.577 - 0.374\, z_\rho(\svec) \), where \( z_\rho(\svec) = \sin(4 s_2 + s_1) - \frac{1}{2} \exp(-64 s_1^2) \). The covariate surfaces are illustrated in Figure~\ref{fig:cov} of the SM.

Model estimation is conducted using a $\MCMC$ sampler with 15{,}000 iterations, burn-in of 7{,}000, and a thinning factor of 8.   Each dataset is fitted using the true copula and the two alternatives. The $\MCMC$ sampler
is implemented in \textsf{R}. For the mesh construction and \(\GMRF\) precision matrix calculations, \textsf{R-INLA} \citep{Bak2018} was used.

\paragraph{Results}
In the SM \ref{apx:sim}, we report in the Table \ref{tab:perc_DIC} the percentages of each selected model under different copula, sample size and type of dependence specification, while Figures~\ref{fig:DIC_N}, \ref{fig:DIC_C} and \ref{fig:DIC_G} in SM~\ref{apx:sim}, report the \(\DIC\) values of the correctly specified model against its competitors. The corresponding $\WAIC$ results are not too far from $\DIC$, and are shown . Points above the diagonal indicate a preference for the true model.

 If the Gaussian copula is the correct model (Figure \ref{fig:DIC_N}), the \(\DIC\) is able to select the true copula for all sample sizes for varying values of the correlation parameter. However, for constant correlation, it is more difficult for the DIC to decide between the Clayton and the Gaussian copula, as in some replications
the Gumbel model yields smaller $\DIC$ values compared to the true model.
For increasing sample sizes, the problem vanishes.
      
If the Clayton copula is the correct model (Figure \ref{fig:DIC_C}), the \(\DIC\) consistently favors the true copula under both constant and varying dependence, with robustness across all sample sizes. The greatest differences in favor of the correct model are observed when the competitor is the Gumbel copula, which is easily explainable due to their different tail dependencies.
    
If the Gumbel copula is the correct model (Figure \ref{fig:DIC_G}), even at smaller sample sizes, the \(\DIC\) reliably selects the correct model. For a small sample size $n=250$, only in a few replications for varying $\rho$, the Gaussian copula yields slightly smaller $\DIC$ values. Differences are larger when the Clayton copula is incorrectly assumed, because of its different tail dependencies.
    
In conclusion, the \(\DIC\) and $\WAIC$ (SM \ref{apx:sim}) are generally effective in identifying the appropriate copula family. The asymmetric copulas, such as Gumbel and Clayton, are more easily distinguished than the Gaussian copula due to their tail properties. The latter is slightly more challenging to identify in scenarios with constant or low correlation and 
small sample sizes.


\section{Application to wind behaviour in Germany}
\label{s:appl}

We apply the proposed PWC model to wind data in Germany. Wind direction, measured in radians on $[0,2\pi)$, indicates the angle of origin. 
Wind speed, recorded in m/s, typically ranges in $(0,31.5]$, under non-hurricane conditions, according to the Beaufort Wind Scale. German wind patterns are mainly shaped by prevailing westerlies from the Atlantic Ocean, with additional easterly currents entering from Poland and the Czech Republic. Their interaction generates variability in northern Germany, whereas the complex orography in central and southern regions (e.g., the Harz Mountains, Black Forest, Bavarian Alps) produces local wind systems and turbulence. Coastal and northern plains typically exhibit higher wind speeds, reflecting the spatial constraints imposed by geography. Given the periodic and spatially correlated nature of the data, we apply the PWC model to analyze wind direction and speed jointly. This approach captures conditional dependence, incorporates covariates in the entire distribution, and accounts for spatial effects via nonstationary Matérn~fields.

\paragraph{Data}
The data used for this application are obtained from the Deutscher Wetterdienst (DWD), and consist of hourly mean wind observations from $289$ weather stations distributed across Germany. This dense network spans diverse geographic settings, enabling the assessment of both large-scale atmospheric dynamics and local effects. Alongside wind speed and direction, additional covariates are available, including visibility range, soil temperature (at a 5cm depth), vapor pressure, air pressure, wet bulb temperature, air temperature (at 2m), relative humidity, lower boundary height, and station altitude. These variables provide valuable explanatory information for structured additive predictors in our model. During storm periods, wind directions tend to be relatively homogeneous across the domain, whereas in calm weather they exhibit greater variability and more frequent directional shifts.
We focus on a stormy weather episode between January 24th and 29th, 2025, and consider, for each station, the circular mean wind direction and the mean wind speed derived from the hourly data. 
As shown in Figure~\ref{f:wind_jan}, the average wind direction was predominantly from the south and west, with a few locations across the country exhibiting mean wind speeds above $10$~m/s. The wind rose confirms the dominance of westerly winds, with limited directional variability. 
These patterns are consistent with the passage of large-scale Atlantic frontal systems, which typically generate strong, persistent westerly flows across Central Europe during storm events.
\begin{figure}[ht]
  \centering
  \begin{minipage}[c]{0.48\textwidth}
    \includegraphics[
    width=\linewidth]{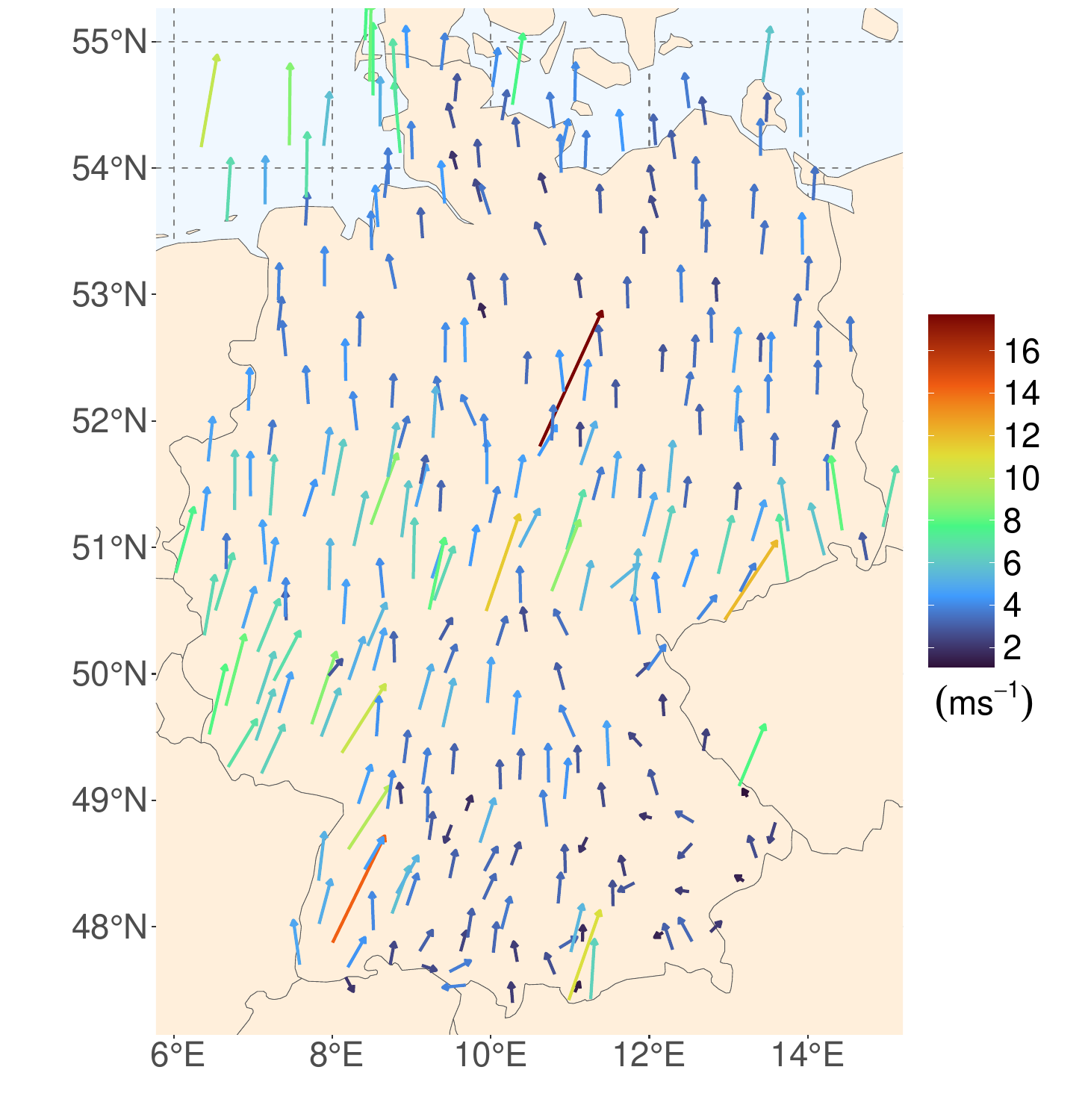}
  \end{minipage}\hfill
  \begin{minipage}[c]{0.48\textwidth}
    \centering
    \includegraphics[
    width=0.9\linewidth]{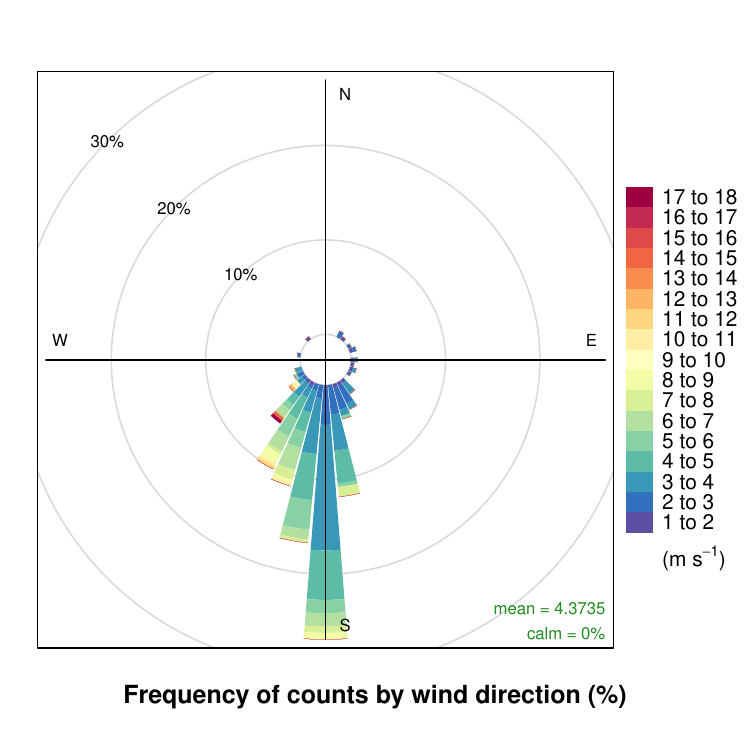}
  \end{minipage}
  \caption{Wind behaviour in Germany during a storm weather period (24-–29 January 2025). Left: station-wise average wind direction and mean wind speed (m/s), with arrow orientation indicating direction and colour scale representing speed. Right: wind rose summarizing the frequency distribution of directions the winds blew from across all stations, with colours denoting wind speed classes.}
  \label{f:wind_jan}
\end{figure}
\paragraph{Model specification}
To investigate the relation between wind direction and speed for these data, we fit our PWC model with wrapped normal and lognormal margins as in \eqref{eq:wn} and \eqref{eq:ln}, by initially including all the covariates for both the two marginal models and for the copula. In Figure \ref{fig:qqplot}  of the SM, we show randomized quantile residuals via the inverse CDF of a standard normal distribution, suggested by \citet{Kle2015} as a simple and effective diagnostic to evaluate the fit of the marginal distributions. If the models are correctly specified, the residuals should follow approximately a standard normal distribution in Figure \ref{fig:qqplot} of the SM.

For wind speed, the log-normal model estimated in the first stage provides an excellent fit to the data. For the circular variable wind direction, the wrapped normal distribution offers a reasonable fit, although it shows some difficulty in capturing the overall shape of the distribution, resulting in slight sigmoidal deviations from the diagonal.
Nevertheless, note that the reported values for the circular distribution are only an approximation of the estimated PDF, for specific values of winding numbers, and the tails could be affected by this approximation.
Thus, we retain the wrapped normal margin, recalling that our focus lies on the benefits of including a covariate-dependent dependence structure into the model rather than assuming independent marginals, and on exploring alternative forms of dependence beyond the symmetric dependence and linear correlation.

Selection of covariates in the marginal models is guided by the $\DIC$ and $\WAIC$, comparing different predictor specifications \citep{Kle2016a}. In the second stage, we analogously examine the role of covariates in the dependence structure induced by the copula. Specifically, for each copula family, constant and covariate-dependent specifications for the copula parameter are compared using $\DIC$ and $\WAIC$.
In particular, we compare the Gaussian copula (N), which exhibits no tail dependence, with the Clayton copula (C), which captures lower-tail dependence, and the Gumbel copula (G), which captures upper-tail dependence. In addition, we assess the sensitivity of dependence modeling under possible misspecification.

\paragraph{Predictive performance}
To further compare competing copula specifications in terms of copula family and in terms of constant or covariate-dependent copula parameter,
We assess predictive performances for the models after variable selection based on $\DIC$ and $\WAIC$.
We use proper scoring rules based on ten-fold cross-validation, with observations randomly assigned to folds. 
We consider the logarithmic score (nLS), the energy score (ES), and a cylindrical definition of the continuous ranked probability score (\cylCRPS), which extends the ES to the cylindrical domain. Lower scores indicate higher predictive accuracy.

\subsection{Results}

\paragraph{Selecting the dependence structure}
The model comparison results in Table~\ref{tab:eval} confirm that accounting for dependence substantially improves the model fit. Both $\DIC$ and $\WAIC$ favor copula-based models over the independence assumption, with further improvement when covariate effects are introduced in the copula parameter. Probabilistic forecasts are assessed through the scoring rules $\nLS$, $\ES$, and $\cylCRPS$. In line with the criteria, including covariates in the dependence structure, the model's scores decrease, thereby delivering a better forecast. 
The Gumbel copula with covariate-dependent parameter (\texttt{G1}) yields the lowest information criteria and best predictive scores across all metrics, indicating the presence of upper-tail dependence between wind direction and speed in this specific storm episode. The initially observed lower-tail pattern under the Clayton copula disappears once covariates are included, suggesting that it was primarily explained by the effects of meteorological covariates rather than true lower-tail asymmetry. Consequently, the subsequent results presented are based on the best performing model G1. 

\begin{table}[htbp]
 \caption{Model comparison based on information criteria (computed based on the complete dataset, using the selected covariates) and average predictive scores (based on ten-fold cross-validation) for the independent model (I), copula models with constant parameter (0), and covariate-dependent parameter (1).  Lower values indicate better performance.
}
 \label{tab:eval}
\centering
 \def\~{\hphantom{0}}
 \begin{tabular}{lccccccc}
   \hline
   Model    & DIC & WAIC & nLS & ES& $\mbox{CRPS}_{\mathit{cyl}}$ & AS    & RMSE  \\ 
   \hline
   I  & 218.460 & 223.390 & 2.773 & 1.211 & 0.440 & 0.183 & 1.222 \\
   N0 & 208.080 & 214.490 & 2.821 & 1.211 & 0.493 & 0.183 & 1.223 \\
   C0 & 206.890 & 212.450 & 2.773 & 1.211 & 0.493 & 0.183 & 1.223 \\
   G0 & 210.830 & 217.500 & 2.802 & 1.211 & 0.493 & 0.182 & 1.223 \\
   N1 & 197.550 & 203.772 & 2.793 & 1.181 & 0.478 & 0.172 & 1.202 \\
   C1 & 195.172 & 200.085 & 2.723 & 1.203 & 0.480 & 0.173 & 1.220 \\
   G1 & \textbf{190.656} & \textbf{195.691} & \textbf{2.532} & \textbf{1.041} & \textbf{0.391} & \textbf{0.097} & \textbf{1.088} \\
   \hline
 \end{tabular}
\end{table}

\paragraph{Estimated posterior effects}
Posterior mean estimates and $95\%$ credible intervals for all G1 model coefficients are reported in Table~\ref{tab:phi-y-rho-summary}.
\begin{table}[ht!]
\centering
\caption{Posterior means and 95\% credible intervals (in square brackets) for the marginal and copula parameters in G1 model.}
\begin{tabular}{lccc}
\toprule
 & $\vartheta_1$ (direction) & $\vartheta_2$ (speed) & $\vartheta_\rho$ (dependence) \\
\midrule
\multirow{2}{*}{Nugget variance} & $0.047$ & $0.072$ & \multicolumn{1}{c}{---} \\
 & $[0.029,\ 0.073]$ & $[0.058,\ 0.088]$ & \multicolumn{1}{c}{---} \\
\midrule
\multirow{2}{*}{Intercept} & $3.216$  & $1.389$  & $-1.560$ \\
 & $[3.192,\ 3.242]$ & $[1.358,\ 1.420]$ & $[-2.292,\ -1.005]$ \\
\multirow{2}{*}{Visibility range} & --- & $-0.125$ & $1.212$ \\
 & --- & $[-0.173,\ -0.077]$ & $[0.308,\ 1.951]$ \\
\multirow{2}{*}{Soil temperature (5 cm)} & $0.011$  & $0.277$  & --- \\
 & $[-0.040,\ 0.063]$ & $[0.200,\ 0.348]$ & --- \\
\multirow{2}{*}{Vapor pressure} & --- & $-0.195$ & --- \\
 & --- & $[-0.337,\ -0.052]$ & --- \\
\multirow{2}{*}{Air pressure} & $-0.180$ & $-0.342$ & $2.375$ \\
 & $[-0.244,\ -0.116]$ & $[-0.399,\ -0.283]$ & $[1.316,\ 3.464]$ \\
\multirow{2}{*}{Air temperature (2 m)} & $0.192$  & $0.292$  & --- \\
 & $[0.107,\ 0.271]$ & $[0.132,\ 0.449]$ & --- \\
\multirow{2}{*}{Relative humidity} & $-0.049$ & $0.211$  & --- \\
 & $[-0.113,\ 0.012]$ & $[0.085,\ 0.337]$ & --- \\
\multirow{2}{*}{Lower boundary height} & $-0.182$ & --- & $-0.915$ \\
 & $[-0.221,\ -0.146]$ & --- & $[-1.492,\ -0.318]$ \\
\multirow{2}{*}{Wet-bulb temperature} & --- & --- & $-1.713$ \\
 & --- & --- & $[-2.811,\ -0.489]$ \\
\midrule
\multicolumn{4}{c}{GMRF log-linear coefficients} \\
\midrule
\multirow{2}{*}{Intercept $\theta_{\tau,0}$} & $-6.543$ & $-2.377$ & --- \\
 & $[-7.397,\ -5.995]$ & $[-2.870,\ -1.779]$ & --- \\
\multirow{2}{*}{Intercept $\theta_{\kappa,0}$} & $5.074$  & $2.595$  & --- \\
 & $[4.674,\ 5.572]$ & $[1.783,\ 3.193]$ & --- \\
\multirow{2}{*}{Altitude $\theta_{\kappa,1}$} & $-0.914$ & $0.047$  & --- \\
 & $[-1.235,\ -0.665]$ & $[-0.256,\ 0.332]$ & --- \\
\bottomrule
\end{tabular}
\label{tab:phi-y-rho-summary}
\end{table}
The coefficients show consistent meteorological patterns: higher air pressure is associated with lower wind speeds and altered wind directions, while higher air temperatures and relative humidity tend to increase wind speeds. The copula parameter $\rho$ shows a positive association with air pressure and visual range, and a negative association with wet-bulb temperature and lower cloud boundary height, indicating stronger dependence under clearer and more stable atmospheric conditions. Spatial random-field parameters (bottom part of the table) reveal marked non-stationarity for the directional component.

\paragraph{Results for spatial prediction}
Spatial prediction results are illustrated in Figure~\ref{fig:pred_wind}. The left panel displays the Cartesian product of the marginal credible intervals (here corresponding to circular sectors) for predicted wind vectors during a storm event at test set locations.
 The joint credible intervals are strictly included in them. The rainbow arrows represent observed wind direction and speed, while the purple and violet arrows denote the $0.025$ and $0.975$ posterior quantiles, respectively. Most observed wind vectors lie within the predicted intervals, demonstrating reliable uncertainty quantification and good calibration. The model successfully captures the wind dynamics across northern and central Germany. In contrast, predictive performance deteriorates in the southern region, particularly near the Alps, where abrupt changes in wind patterns and orographic effects are more complex to capture with the current mesh resolution. A finer mesh in these areas would improve local spatial representation and reduce predictive bias.
The right panel of Figure~\ref{fig:pred_wind} shows posterior mean estimates of the copula parameter $\rho$, which exhibits a clear spatial gradient, with higher values in the north indicating stronger dependence between direction and speed. In comparison, lower values in the mountainous south reflect weaker or more variable associations, consistent with the increased turbulence typical of complex terrain.

\begin{figure}[h!]
    \centering
 \begin{minipage}[c]{0.48\textwidth}    
\includegraphics[width=0.95\linewidth]{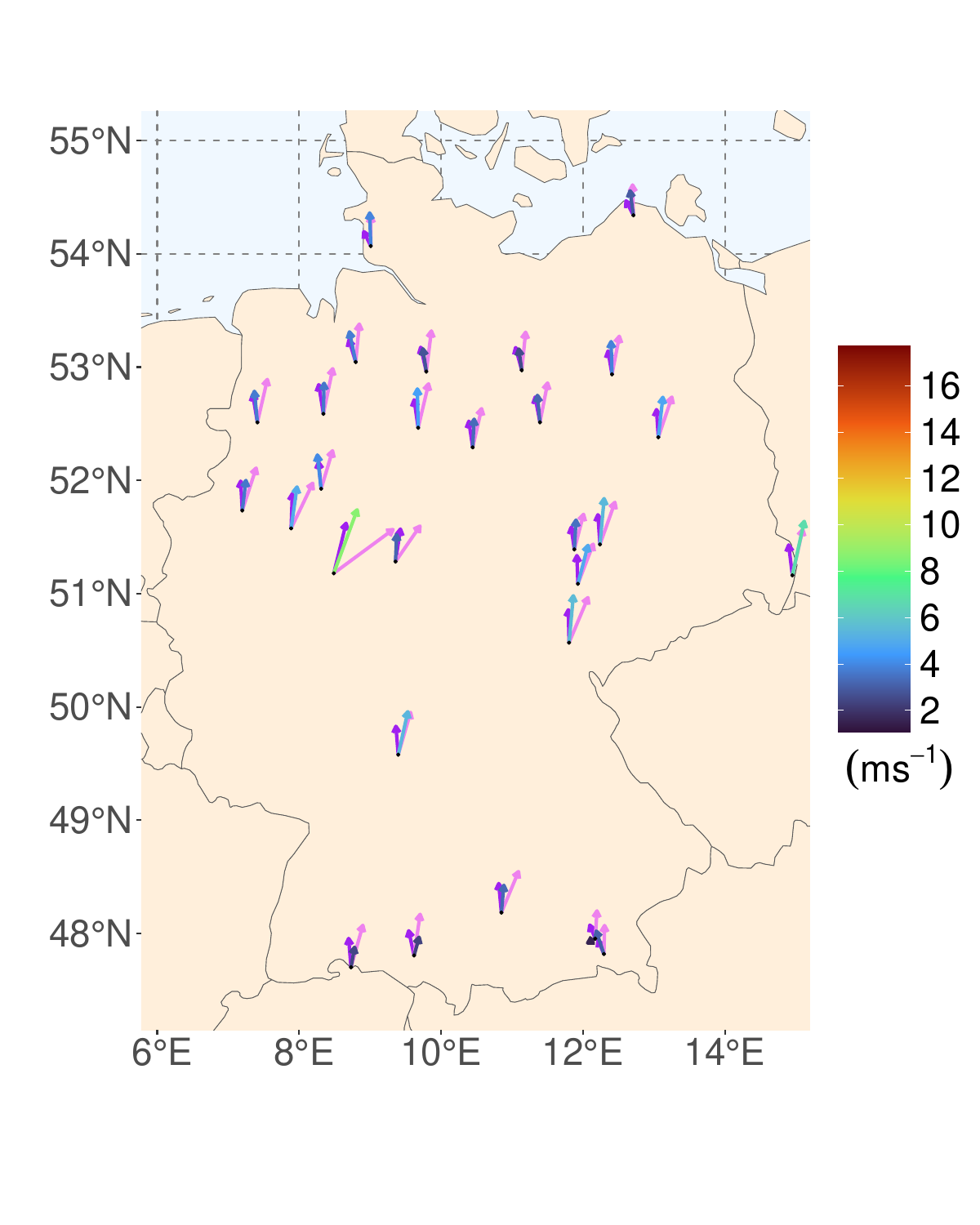}
\end{minipage}
 \begin{minipage}[c]{0.48\textwidth}
\includegraphics[width=0.9\linewidth]{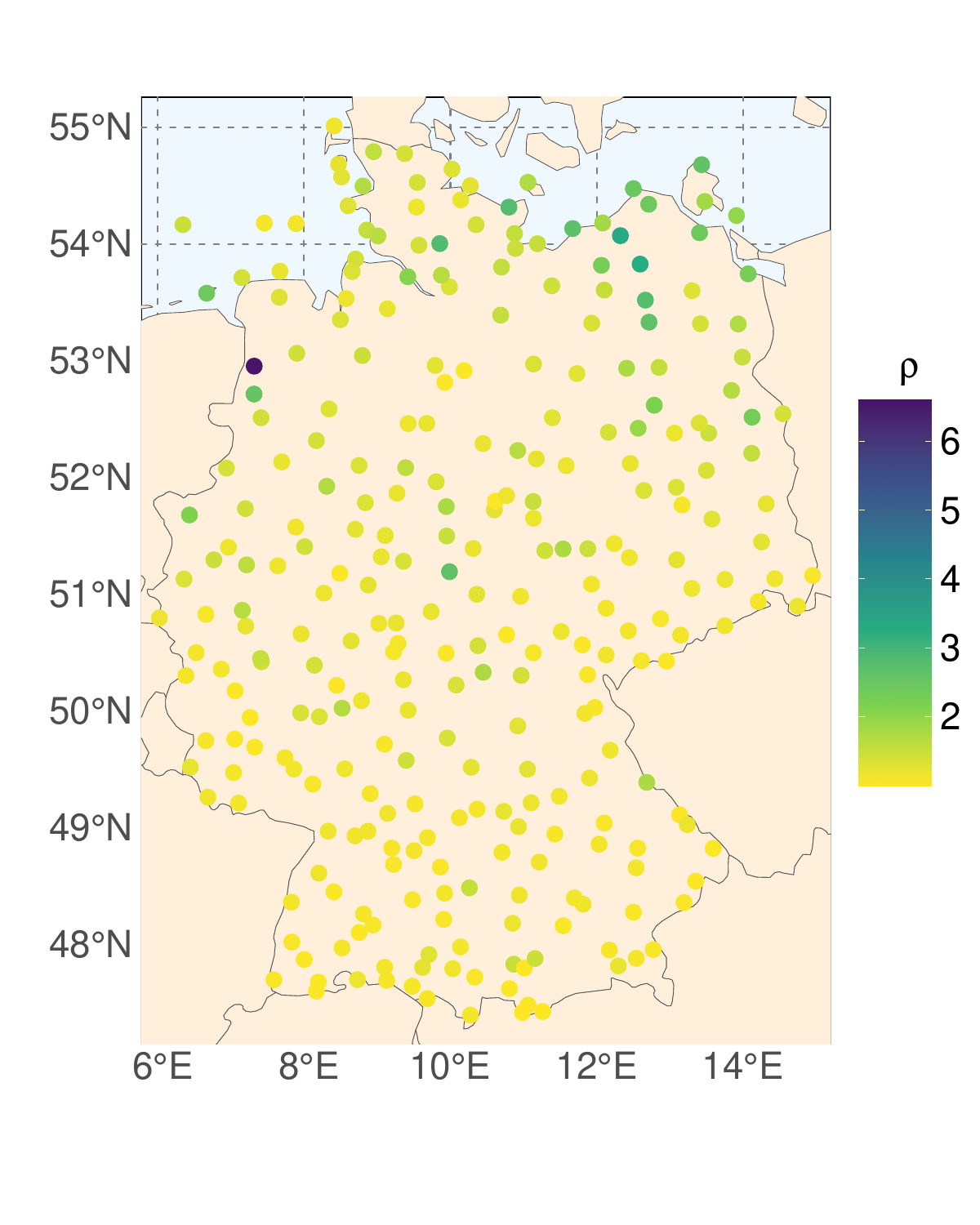}
 \end{minipage}
    \caption{Left: credible intervals for wind vectors at test locations during a storm period. The rainbow arrows indicate the observed wind direction and speed, while the purple and violet arrows denote the $0.025$ and $0.975$ posterior quantiles, respectively. Right: posterior mean estimates of the copula parameter $\rho$ at all locations.}
    \label{fig:pred_wind}
\end{figure}

\section{Conclusion and outlook}
\label{s:discuss}

We proposed the PWC model, a generalization of copula-based models for directional statistics that allows for flexible specification of the linear components and recovers the cylindrical joint model through a wrapping construction. The PWC model accommodates spatially correlated data by defining marginal models driven by non-stationary latent GRFs, and introduces a conditional copula regression framework linking the specified marginals. Additionally, we investigated the role of covariates in the dependence structure, 
showing that accounting for environmental variables in the association between wind speed and wind direction can substantially improve model~fit.

Several extensions of the proposed framework are possible. For the marginal distributions, more flexible formulations may be adopted, such as the wrapped skew-normal \citep{Mas2016wsgp} or wrapped mixtures \citep{
Gre2023} for the circular component, and the generalized gamma 
or Dagum 
distributions for the linear component. Furthermore, the copula governing the dependence structure could be made location-specific or replaced with alternatives based on trigonometric functions or semi-parametric formulations \citep{Kau2013,Kus2025} to capture more complex forms of association.
Future work will focus on extending the model to a fully spatio-temporal framework, allowing dynamic dependence between circular and linear processes to be modeled over space and time.

\bibliographystyle{chicago} 
\bibliography{bib.bib}

@article{Aca2011,
  title={Dependence calibration in conditional copulas: a nonparametric approach},
  author={Acar, E. F. and Craiu, Radu V. and Yao, F.},
  journal={Biometrics},
  volume={67},
  number={2},
  pages={445--453},
  year={2011},
  publisher={Oxford University Press}
}

@article{Bak2018,
  author  = {Bakka, H. and Rue, H. and Fuglstad, G.-A. and Riebler, A. and Bolin, D. and Illian, J. and Krainski, E. and Simpson, D. and Lindgren, F.},
  title   = {Spatial modeling with {R-INLA}: a review},
  journal = {Wiley Interdisciplinary Reviews: Computational Statistics},
  year    = {2018},
  volume  = {10},
  number  = {6},
  pages   = {e1443}
}

@article{Bol2019,
  author  = {Bolin, D. and Kirchner, K.},
  title   = {The rational {SPDE} approach for {G}aussian random fields with general smoothness},
  journal = {Journal of Computational and Graphical Statistics},
  year    = {2020},
  volume  = {29},
  number  = {2},
  pages   = {274--285}
}

@article{Car2009,
  author  = {Carta, J. A. and Ramirez, P. and Velazquez, S.},
  title   = {A review of wind speed probability distributions used in wind energy analysis: case studies in the {C}anary {I}slands},
  journal = {Renewable and Sustainable Energy Reviews},
  year    = {2009},
  volume  = {13},
  number  = {5},
  pages   = {933--955}
}

@article{Dur2007,
    author = {Fernández-Durán, J. J.},
    title = {Models for Circular–Linear and Circular–Circular Data Constructed from Circular Distributions Based on Nonnegative Trigonometric Sums},
    journal = {Biometrics},
    volume = {63},
    number = {2},
    pages = {579-585},
    year = {2007},
    month = {01}
}

@Book{Fah2022,
  author    = {Fahrmeir, L. and Kneib, T. and Lang, S. and Marx, B. D.},
  title     = {Regression: Models, Methods and Applications},
  publisher = {Springer},
  address   = {Berlin Heidelberg},
  year      = {2022}
}

@article{Gel1997,
  author  = {Gelman, A. and Gilks, W. R. and Roberts, G. O.},
  title   = {Weak convergence and optimal scaling of random walk {M}etropolis algorithms},
  journal = {Annals of Applied Probability},
  year    = {1997},
  volume  = {7},
  number  = {1},
  pages   = {110--120}
}

@Book{Gel2013,
  author    = {Gelman, A. and Stern, H. S. and Carlin, J. B. and Dunson, D. B. and Vehtari, A. and Rubin, D. B.},
  title     = {Bayesian Data Analysis},
  edition   = {3rd},
  publisher = {Chapman \& Hall/CRC},
  address   = {Boca Raton, FL},
  year      = {2013}
}

@article{Gre2023,
  title={Finite mixtures of multivariate wrapped normal distributions for model based clustering of p-torus data},
  author={Greco, L. and Novi Inverardi, P. L. and Agostinelli, C.},
  journal={Journal of Computational and Graphical Statistics},
  volume={32},
  number={3},
  pages={1215--1228},
  year={2023},
  publisher={Taylor \& Francis}
}

@article{Hod2022,
  title={Circular--linear copulae for animal movement data},
  author={Hodel, F. H. and Fieberg, J. R.},
  journal={Methods in Ecology and Evolution},
  volume={13},
  number={5},
  pages={1001--1013},
  year={2022},
  publisher={Wiley Online Library}
}

@article{Ing2015,
  author  = {Ingebrigtsen, R. and Lindgren, F. and Steinsland, I. and Martino, S.},
  title   = {Estimation of a non-stationary model for annual precipitation in southern {N}orway using replicates of the spatial field},
  journal = {Spatial Statistics},
  year    = {2015},
  volume  = {14},
  pages   = {338--364}
}

@article{Joh1978,
  title={Some angular-linear distributions and related regression models},
  author={Johnson, R. A. and Wehrly, T. E.},
  journal={Journal of the American Statistical Association},
  volume={73},
  number={363},
  pages={602--606},
  year={1978},
  publisher={Taylor \& Francis}
}

@article{Jon2015,
  title={On a class of circulas: copulas for circular distributions},
  author={Jones, M. C. and Pewsey, A. and Kato, S.},
  journal={Annals of the Institute of Statistical Mathematics},
  volume={67},
  number={5},
  pages={843--862},
  year={2015},
  publisher={Springer}
}

@article{Kau2013,
  title={Flexible copula density estimation with penalized hierarchical {B}-splines},
  author={Kauermann, G. and Schellhase, C. and Ruppert, D.},
  journal={Scandinavian Journal of Statistics},
  volume={40},
  number={4},
  pages={685--705},
  year={2013},
  publisher={Wiley Online Library}
}

@article{Kle2015,
  author  = {Klein, N. and Kneib, T. and Lang, S. and Sohn, A.},
  title   = {Bayesian structured additive distributional regression with an application to regional income inequality in {G}ermany},
  journal = {The Annals of Applied Statistics},
  year    = {2015},
  volume  = {9},
  number  = {2},
  pages   = {1024--1052}
}

@article{Kle2016a,
  author  = {Klein, N. and Kneib, T.},
  title   = {Simultaneous inference in structured additive conditional copula regression models: a unifying {B}ayesian approach},
  journal = {Statistics and Computing},
  year    = {2016},
  volume  = {26},
  pages   = {841--860}
}

@article{Kle2016b,
  author  = {Klein, N. and Kneib, T.},
  title   = {Scale-dependent priors for variance parameters in structured additive distributional regression},
  journal = {Bayesian Analysis},
  year    = {2016},
  volume  = {11},
  pages   = {1071--1106}
}

@InProceedings{Kur2014b,
  author    = {Kurz, G. and Gilitschenski, I. and Hanebeck, U. D.},
  title     = {Efficient evaluation of the probability density function of a wrapped normal distribution},
  booktitle = {2014 Sensor Data Fusion: Trends, Solutions, Applications (SDF)},
  publisher = {IEEE},
  year      = {2014},
  pages     = {1--5}
}

@article{Kus2025,
  title={Grid-Uniform Copulas and Rectangle Exchanges: {B}ayesian Model and Inference for a Rich Class of Copula Functions},
  author={Kuschinski, N. and Jara, A.},
  journal={Bayesian Analysis},
  volume={20},
  number={1},
  pages={55--82},
  year={2025},
  publisher={International Society for Bayesian Analysis}
}

@incollection{Lag2018,
  title={Correlated cylindrical data},
  author={Lagona, Francesco},
  booktitle={Applied Directional Statistics},
  pages={61--76},
  year={2018},
  publisher={Chapman and Hall/CRC}
}

@article{Lag2019,
  title={Copula-based segmentation of cylindrical time series},
  author={Lagona, Francesco},
  journal={Statistics \& Probability Letters},
  volume={144},
  pages={16--22},
  year={2019},
  publisher={Elsevier}
}

@article{Lag2025,
  title={Nonhomogeneous hidden semi-{M}arkov models for toroidal data},
  author={Lagona, F. and Mingione, M.},
  journal={Journal of the Royal Statistical Society Series C: Applied Statistics},
  volume={74},
  number={1},
  pages={142--166},
  year={2025},
  publisher={Oxford University Press UK}
}

@article{Lan2019,
  author  = {Lang, M. N. and Mayr, G. J. and Stauffer, R. and Zeileis, A.},
  title   = {Bivariate {G}aussian models for wind vectors in a distributional regression framework},
  journal = {Advances in Statistical Climatology, Meteorology and Oceanography},
  year    = {2019},
  volume  = {5},
  number  = {2},
  pages   = {115--132}
}

@article{Jon2012,
  author  = {Jona-Lasinio, G. and Gelfand, A. and Jona-Lasinio, M.},
  title   = {Spatial analysis of wave direction data using wrapped {G}aussian processes},
  journal = {The Annals of Applied Statistics},
  year    = {2012},
  volume  = {6},
  number  = {4},
  pages   = {1478--1498}
}

@article{Jon2020,
  author  = {Jona-Lasinio, G. and Santoro, M. and Mastrantonio, G.},
  title   = {CircSpaceTime: an {R} package for spatial and spatio-temporal modelling of circular data},
  journal = {Journal of Statistical Computation and Simulation},
  year    = {2020},
  volume  = {90},
  number  = {7},
  pages   = {1315--1345}
}

@article{Lin2011,
  author  = {Lindgren, F. and Rue, H. and Lindstr{\"o}m, J.},
  title   = {An explicit link between {G}aussian fields and {G}aussian {M}arkov random fields: the stochastic partial differential equation approach},
  journal = {Journal of the Royal Statistical Society: Series B},
  year    = {2011},
  volume  = {73},
  number  = {4},
  pages   = {423--498}
}

@article{Mas2016wsgp,
  author  = {Mastrantonio, G. and Gelfand, A. E. and Jona-Lasinio, G.},
  title   = {The wrapped skew {G}aussian process for analyzing spatio-temporal data},
  journal = {Stochastic Environmental Research and Risk Assessment},
  year    = {2016},
  volume  = {30},
  number  = {8},
  pages   = {2231--2242}
}

@Book{Mar2000,
  author    = {Mardia, K. V. and Jupp, P. E.},
  title     = {Directional Statistics},
  publisher = {John Wiley \& Sons},
  address   = {Chichester},
  year      = {2000}
}

@article{Mar2022,
  author  = {Marques, I. and Kneib, T. and Klein, N.},
  title   = {A non-stationary model for spatially dependent circular response data based on wrapped {G}aussian processes},
  journal = {Statistics and Computing},
  year    = {2022},
  volume  = {32},
  number  = {5},
  pages   = {73}
}

@article{Mas2022,
  title={Modeling animal movement with directional persistence and attractive points},
  author={Mastrantonio, Gianluca},
  journal={The Annals of Applied Statistics},
  volume={16},
  number={3},
  pages={2030--2053},
  year={2022},
  publisher={Institute of Mathematical Statistics}
}

@article{Mei2021,
  title={Nonparametric estimation of circular trend surfaces with application to wave directions},
  author={Meil{\'a}n-Vila, A. and Crujeiras, R. M. and Francisco-Fern{\'a}ndez, M.},
  journal={Stochastic Environmental Research and Risk Assessment},
  volume={35},
  number={4},
  pages={923--939},
  year={2021},
  publisher={Springer}
}

@article{Mil2020,
  author  = {Miller, D. L. and Glennie, R. and Seaton, A. E.},
  title   = {Understanding the stochastic partial differential equation approach to smoothing},
  journal = {Journal of Agricultural, Biological and Environmental Statistics},
  year    = {2020},
  volume  = {25},
  number  = {1},
  pages   = {1--16}
}

@Book{Nel2006,
 author    = {Nelsen, R.},
  title     = {An Introduction to Copulas},
  publisher = {Springer},
 address   = {Portland},
  year      = {2006}
}

@incollection{Oga2023,
  title={Copula bounds for circular data},
  author={Ogata, Hiroaki},
  booktitle={Research Papers in Statistical Inference for Time Series and Related Models: Essays in Honor of Masanobu Taniguchi},
  pages={389--402},
  year={2023},
  publisher={Springer}
}

@article{Pew2021,
  author  = {Pewsey, A. and Garc{\'\i}a-Portugu{\'e}s, E.},
  title   = {Recent advances in directional statistics},
  journal = {TEST},
  year    = {2021},
  volume  = {30},
  pages   = {1--58}
}

@article{Sim2017,
  author  = {Simpson, D. and Rue, H. and Riebler, A. and Martins, T. G. and S{\o}rbye, S. H.},
  title   = {Penalising model component complexity: a principled, practical approach to constructing priors},
  journal = {Statistical Science},
  year    = {2017},
  volume  = {32},
  number  = {1},
  pages   = {1--28}
}

@article{Spi2002,
  author  = {Spiegelhalter, D. J. and Best, N. G. and Carlin, B. P. and van der Linde, A.},
  title   = {Bayesian measures of model complexity and fit},
  journal = {Journal of the Royal Statistical Society: Series B},
  year    = {2002},
  volume  = {64},
  number  = {4},
  pages   = {583--639}
}

@article{Tan2021,
  author  = {Tang, W. and Zhang, L. and Banerjee, S.},
  title   = {On identifiability and consistency of the nugget in {G}aussian spatial process models},
  journal = {Journal of the Royal Statistical Society: Series B },
  year    = {2021},
  volume  = {83},
  number  = {5},
  pages   = {1044--1070}
}

@article{Vat2018,
  title={Generalized additive models for pair-copula constructions},
  author={Vatter, T. and Nagler, T.},
  journal={Journal of Computational and Graphical Statistics},
  volume={27},
  number={4},
  pages={715--727},
  year={2018},
  publisher={Taylor \& Francis}
}

@article{Veh2017,
  author  = {Vehtari, A. and Gelman, A. and Gabry, J.},
  title   = {Practical {B}ayesian model evaluation using leave-one-out cross-validation and {WAIC}},
  journal = {Statistics and Computing},
  year    = {2017},
  volume  = {27},
  number  = {5},
  pages   = {1413--1432}
}

@article{Vih2012,
  author  = {Vihola, M.},
  title   = {Robust adaptive {M}etropolis algorithm with coerced acceptance rate},
  journal = {Statistics and Computing},
  year    = {2012},
  volume  = {22},
  number  = {5},
  pages   = {997--1008}
}

@article{Wat2010,
  author  = {Watanabe, S. and Opper, M.},
  title   = {Asymptotic equivalence of {B}ayes cross validation and widely applicable information criterion in singular learning theory},
  journal = {Journal of Machine Learning Research},
  year    = {2010},
  volume  = {11},
  pages   = {3571--3594}
}

@book{Woo2017,
  title={Generalized additive models: an introduction with {R}},
  author={Wood, S. N.},
  year={2017},
  publisher={Chapman \& Hall/CRC}
}

\clearpage
\begin{center} \LARGE 
 \bf{Supplementary Material}   
\end{center}

\setcounter{section}{0}
\renewcommand{\thesection}{\Alph{section}}
\renewcommand{\thesubsection}{\thesection.\arabic{subsection}} 
\renewcommand{\thefigure}{S\arabic{figure}}
\renewcommand{\thetable}{S\arabic{table}}

\section{Theoretical results and proofs}\label{app:proofs}

\subsection{Details on $p_{1,2}^{\mathrm{PW}}$}\label{apx:PWNmarg}

The joint density \( p_{1,2}^{\mathrm{PW}}(\varphi_1, y_2 \mid \bm{z}) \) is well-defined as it satisfies the standard requirements for a valid probability density function.

First, the total integral over the cylindrical domain equals one:
\[
\int_{-\infty}^{+\infty} \int_0^{2\pi} p_{1,2}^{\mathrm{PW}}(\varphi_1, y_2 \mid \bm{z}) \, d\varphi_1 \, dy_2 = 1,
\]
which follows directly from Lemma~\ref{lem:pw} and the fact that the wrapped marginal integrates the unwrapped joint density over the infinite lattice of \( 2\pi \)-shifts.

Second, non-negativity is ensured since \( p_{1,2}^{\mathrm{PW}} \) is defined as a countable sum of non-negative terms, each being the evaluation of a Euclidean joint density:
\[
p_{1,2}^{\mathrm{PW}}(\varphi_1, y_2 \mid \bm{z}) = \sum_{k \in \mathbb{Z}} p_{1,2}(\varphi_1 + 2\pi k, y_2 \mid \bm{z}) \geq 0.
\]

Third, \( p_{1,2}^{\mathrm{PW}} \) satisfies \( 2\pi \)-periodicity in the circular component. For all \( j \in \mathbb{Z} \), we have:
\begin{align*}
p_{1,2}^{\mathrm{PW}}(\varphi_1 + 2\pi j, y_2 \mid \bm{z}) 
&= \sum_{k \in \mathbb{Z}} p_{1,2}(\varphi_1 + 2\pi j + 2\pi k, y_2 \mid \bm{z}) \\
&= \sum_{h \in \mathbb{Z}} p_{1,2}(\varphi_1 + 2\pi h, y_2 \mid \bm{z}) \\
&= p_{1,2}^{\mathrm{PW}}(\varphi_1, y_2 \mid \bm{z}),
\end{align*}
where the index substitution \( h = j + k \) justifies equality.

Hence, \( p_{1,2}^{\mathrm{PW}} \) is a valid cylindrical probability density function.

\subsection{Proof of Lemma \ref{lem:pw}}\label{app:prooflemma}

\begin{proof} Marginalization of the circular part $\int_0^{2\pi} p_{1,2}^{PW}(\varphi_1,y_2\mid \bm{z}) \, d\varphi_1$ yields
 \begin{align*}
     &\; \int_0^{2\pi}\sum_{k \in \mathbb{Z}} c\!\left( F_1(\varphi_1 + 2\pi k \mid \bm{z}),  F_2(y_{2} \mid \bm{z}) \middle| \bm{z} \right) p_1(\varphi_1 + 2\pi k \mid \bm{z})p_2(y_{2} \mid \bm{z}) \, d\varphi_1\\
     &\overset{(a)}{=}\sum_{k \in \mathbb{Z}} \int_0^{2\pi} c\!\left( F_1(\varphi_1 + 2\pi k \mid \bm{z}),  F_2(y_{2} \mid \bm{z}) \middle| \bm{z} \right) p_1(\varphi_1 + 2\pi k \mid \bm{z})p_2(y_2 \mid \bm{z})\, d\varphi_1\\
     &\overset{(b)}{=}\int_{-\infty}^{+\infty} c\!\left( F_1(y_1 \mid \bm{z}),  F_2(y_{2} \mid \bm{z}) \middle| \bm{z} \right) p_1(y_1 \mid \bm{z})p_2(y_2 \mid \bm{z})\, dy_1 \\
     &\overset{(c)}{=}p_2(y_2 \mid \bm{z}).
 \end{align*}
 Marginalization of the linear part $\int_{-\infty}^{\infty} p_{1,2}^{PW}(\varphi_1,y_2\mid \bm{z}) \, dy_2$ yields
 \begin{align*}
     &\; \int_{-\infty}^{\infty} \sum_{k \in \mathbb{Z}} c\!\left( F_1(\varphi_1 + 2\pi k \mid \bm{z}),  F_2(y_{2} \mid \bm{z}) \middle| \bm{z} \right) p_1(\varphi_1 + 2\pi k \mid \bm{z})p_2(y_{2} \mid \bm{z}) \, dy_2\\
     &\overset{(a)}{=}\sum_{k \in \mathbb{Z}} \int_{-\infty}^{\infty}  c\!\left( F_1(\varphi_1 + 2\pi k \mid \bm{z}),  F_2(y_{2} \mid \bm{z}) \middle| \bm{z} \right) p_1(\varphi_1 + 2\pi k \mid \bm{z})p_2(y_2 \mid \bm{z})\, dy_2\\
     &\overset{(c)}{=}\sum_{k \in \mathbb{Z}} p_1(\varphi_1+2\pi k \mid \bm{z}) \\
     &=p_1^W(\varphi_1 \mid \bm{z}).
 \end{align*}
At (a), we use the dominated convergence theorem, at (b) we
use the concatenation of integrals, and at (c) Sklar's theorem.
\end{proof}

\section{Summary on copula specifications} \label{a:Table_copula}

\begin{table}[ht]

  \caption{Considered one-parameter copula functions with their admissible range for the association parameter $\rho$ 
  and the corresponding link function. 
  Here, $\Phi_{2}(\cdot, \cdot; \rho)$ denotes the CDF of a standard bivariate Gaussian distribution 
  with correlation coefficient $\rho$, and $\Phi(\cdot)$ denotes the CDF of a standard univariate Gaussian distribution. 
}
\label{t:cop}
\begin{center}
\begin{tabular}{lccc}
\hline\hline
Copula & \multicolumn{1}{c}{$C(u_1,u_2 \, ; \, \rho )$} &  \multicolumn{1}{c}{Range of $\rho$} & 
\multicolumn{1}{c}{Link function} \\ \hline
Gaussian & $\Phi_{2}\!\left(\Phi^{-1}(u_{1}), \, \Phi^{-1}(u_{2}) \, ; \, \rho\right)$ & $[-1,1]$ & $\dfrac{\rho}{\sqrt{1-\rho^{2}}}$  \\
Clayton  & $\big(u_{1}^{-\rho}+u_{2}^{-\rho}-1\big)^{-1/\rho}$ & $(0, \infty)$ & $\log(\rho 
)$ \\
Gumbel   & $\exp \!\left[-\left\{\big(-\log u_{1}\big)^{\rho}+\big(-\log u_{2}\big)^{\rho}\right\}^{1/\rho}\right]$ & $(1, \infty)$ & $\log(\rho - 1
)$ \\
\hline\hline
\end{tabular}
\end{center}
\end{table}

\section{Details on MCMC sampler}\label{app:MCMC}
\subsection{Full conditionals for Gibbs sampler}\label{apx:fullCon}
Let's consider a vector of Gaussian observations $\bm{y}=(y_1, \ldots, y_n)^\top$, and likelihood with mean parameters $\{\mu_i\}_{i=1}^n$ and variances $\sigma^2$, and a generic mean predictor 
\[
\boldsymbol{\eta}_{\bm{\mu}} = Z_0 \beta_0 + Z_1\boldsymbol{\beta}_1 + \ldots +Z_J\boldsymbol{\beta}_J,
\]
where $Z_0=\bm{1}$ is the unit vector, $\beta_0$ denotes the intercept, $Z_j$ a generic design matrix, related to some effects, and $\bm{\beta}_j$ is the corresponding parameter vector. For all mean coefficients endowed with Gaussian (N) priors (e.g., $\beta_0 \sim N(0, a_0)$ and $\bm{\beta}_j \mid \xi_j^2 \sim N(0\,,\,\xi_j^2\mathbf{I})$) and Inverse Gamma (IG) priors for the variance (e.g, $\sigma^2\sim IG(0.001,0.001)$), full conditional updates can be derived analytically. In particular, the product of two Gaussian densities is proportional to a Gaussian density, which implies conjugacy and leads to closed-form full conditional distributions. \citep[see, e.g.,][]{Fah2022}. Exploiting this result, we obtain the following general full conditional forms:
\[
\boldsymbol{\beta}_j\mid \cdot \sim 
\mathcal{N}\!\left(
  (E + F)^{-1}(E\mathbf{e} + F\mathbf{f}),\;
  (E + F)^{-1}
\right),
\]
$$
\sigma^2 \mid \cdot \sim 
\text{IG}\!\left(
0.001 + \frac{n}{2},
\;0.001 + \frac{1}{2}(\mathbf{y} - \boldsymbol{\eta}_\mu)^\top
(\mathbf{y} - \boldsymbol{\eta}_\mu)
\right).
$$
where $E= \frac{1}{\sigma^2} {Z}_j^\top{Z}_j$ collects the information from the Gaussian likelihood, $F= 1/\xi^{2}\mathbf{I}$ (or $F= 1/a_0$ for the intercept parameter) encodes the corresponding Gaussian prior contribution, $\mathbf{e} = Z_j^\top(Z_j Z_j^\top)^{-1}(\mathbf{Y} - \bm{\eta}_\mu +Z_j\bm{\beta}_j)$ and
$\mathbf{f} = \bm{0}$.

\subsection{Gaussian Markov random field precision matrix}\label{apx:GMRF}

For a stationary $\GRF$ solution to the $\SPDE$, the $\GMRF$ precision is
\[
Q = \tau\left(\kappa^4 C + 2\kappa^2 G + G C^{-1} G\right)\tau,
\]
where $C$ and $G$ are $\FEM$ mass and stiffness matrices, respectively.

With spatially varying $\kappa(\mathbf{s})$, using mass lumping so $C$ is diagonal:
\[
Q = \tau\left(
\operatorname{diag}(\kappa)^2 C \operatorname{diag}(\kappa)^2 +
\operatorname{diag}(\kappa)^2 G +
G' \operatorname{diag}(\kappa)^2 +
G C^{-1}G
\right)\tau.
\]

\section{ Penalized Complexity Priors for variances}\label{apx:PC}

Penalized Complexity (PC) priors \citep{Sim2017} formalize Occam’s razor by shrinking toward a \emph{base model} and penalizing deviations from it at a constant rate. For variance components, the base model is the "no‐variance" model (random–effect variance equal to zero).

\vskip.5cm \paragraph{Occam's razor} A simpler model formulation (base model) should be preferred until there is enough support for a more complex model. 

\vskip.5cm \paragraph{Measure of complexity}
Let $p=p(y\mid \xi)$ denote the likelihood under parameter $\xi$ (e.g., a variance
or standard deviation) and let $p_0=p(y\mid \xi_0)$ be the likelihood under the base model. Define the unidirectional distance from the base model via the
Kullback–Leibler divergence (KLD):
\[
d(\xi) \;=\; \sqrt{2\,\mathrm{KLD}\!\left(p\,\Vert\,p_0)\right)}\,, \text{where }
\,
\mathrm{KLD}(p\Vert q)=\int p(u)\log\!\Big(\tfrac{p(u)}{q(u)}\Big)\,du.
\]
Here $d(\xi)\ge 0$ and $d(\xi_0)=0$. Larger $d$ means greater deviation from the simpler model.

\vskip.5cm \paragraph{Constant rate penalization}
PC priors assume an exponential law on the distance scale,
\[
\pi_d(d) \;=\; \lambda\,\exp(-\lambda d), \qquad d\ge 0,
\]
which yields (by change of variables) the prior on $\xi$:
\[
\pi(\xi) \;=\; \lambda \exp\!\big(-\lambda\,d(\xi)\big)\,
\Bigl|\tfrac{d}{d\xi}d(\xi)\Bigr|.
\]
This places maximal prior mass at the base model and yields exponentially
decreasing weight for more complex models.

\vskip.5cm \paragraph{User–defined scaling (tail calibration)}
The rate $\lambda$ is fixed by a simple tail–probability statement reflecting
a sensible scale for the parameter. Once $U>0$ as a “sensible”, user-defined upper bound that specifies the prior knowledge about “tail events” and $\alpha\in(0,1)$ as the weight to put on this event, are chosen, the condition
$\mathbb{P}\left(Q(\xi)>U\right)=\alpha$ is imposed for a monotone scale function $Q$, which is an interpretable transformation of the flexibility parameter.
This condition allows the user to prescribe how informative the resulting PC prior shall be.

\section{Details on model choice}\label{s:evmet}
The $\DIC$ and $\WAIC$ are simulation-based criteria, which are straightforward to compute from $\MCMC$ output and have been widely used for comparing response distributions and predictor specifications, in a stepwise model choice strategy. The $\DIC$ combines a measure of model fit, given by the average deviance across posterior draws $\overline{D(\vartheta)}$, with a penalty for model complexity, defined as the difference between the mean deviance and the deviance at a representative point of the posterior (e.g., the posterior mean) $D(\bar\vartheta)$.
The criterion is formally defined as
$   \text{DIC}= 
\frac{2}{T} \sum D(\bm{\vartheta}^{[t]})-D(\frac{1}{T} \, \sum\bm{\vartheta}^{[t]})$ where $D(\bm{\vartheta}) = -2 \log p(y \mid \bm{\vartheta}).
$

The $\WAIC$ provides a fully Bayesian alternative, based on the log-pointwise predictive density $\text{lppd} = \sum_{i=1}^n \log \int p(y_i \mid \bm{\vartheta}) \, p(\bm{\vartheta} \mid y) \, d\bm{\vartheta}$ with a complexity penalty given by the variance of the log-likelihood across posterior draws and summed over data points $ p_{\text{WAIC}} = \sum_{i=1}^n 2 \, \text{var}_{\bm{\vartheta}}\!\big( \log p(y_i \mid \bm{\vartheta}) \big),$ (\cite{Gel2013,Veh2017}), yielding to $
   \text{WAIC} = -2 \, (\text{lppd} - p_{\text{WAIC}}).$
The $\WAIC$ is invariant to parameterization, robust in the presence of skewed or multimodal posteriors, while the $\DIC$ may be sensitive when posterior distributions deviate from normality. Since both are sample-based, slight differences between competing models may lead to a region of indecisiveness; however, the criteria have been shown to favor sparser models when combined with the exclusion of effects whose credible intervals include zero.
Simulation studies in \cite{Kle2015} further suggest that the $\DIC$
effectively detect omitted relevant covariates, while the inclusion of irrelevant ones usually results only in insignificant effects.

\section{Simulation study}\label{apx:sim}
We report here the table and figures discussed in Sec. \ref{s:sim}.

\begin{figure}[h!]
    \centering
    \includegraphics[width=1\linewidth]{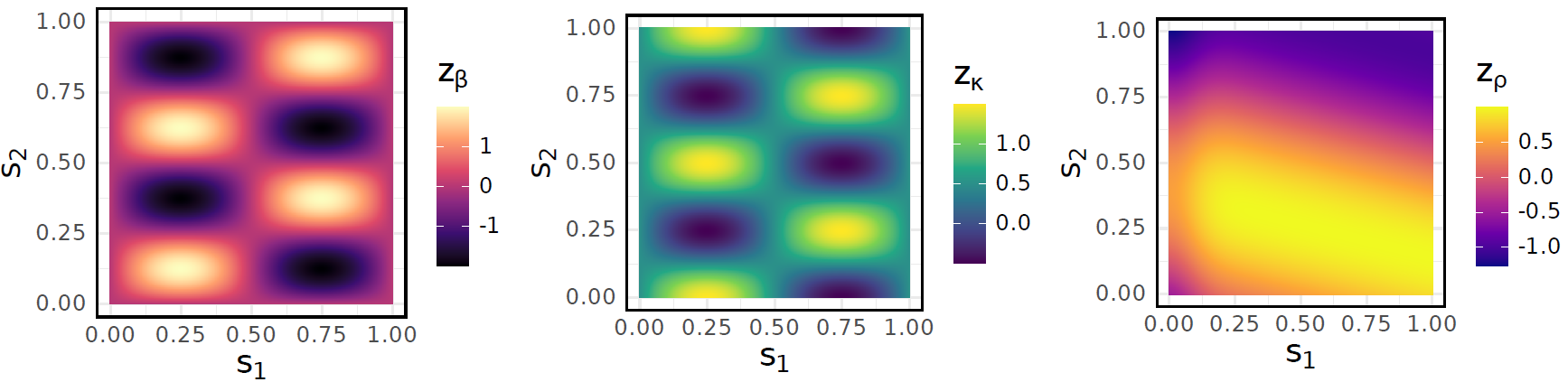}
    \caption{The three boxes show, from left to right, the covariates functions $z_\beta(\svec)=2\sin(2\pi \,s_1)\sin(4\pi\,s_2)$, $z_\kappa(\svec)= 1/2+\sin(2\pi \, s_1)\cos(4\pi\,s_2)$, and $z_\rho(\bm{s})=\sin(4*s_2+s_1)-\frac{1}{2}\exp(-64s_1^2)$ used in the simulation study.}
    \label{fig:cov}
\end{figure}

\begin{table}[h!]
\centering
\caption{Percentage of replications with minimum $\DIC$ among the three investigated copula~models. Columns report combinations of the true copula model under which data were generated and sample size $n$, while the rows show the fitted copula model (N=``Gaussian'', C=``Clayton'', G=``Gumbel''), with a constant or covariate-dependent (varying) copula dependence parameter $\rho$.}
\small
\begin{tabular}{l l rrrrrrrrr}
\toprule
\textbf{DIC} & & \multicolumn{3}{c}{Clayton} & \multicolumn{3}{c}{Gumbel} & \multicolumn{3}{c}{Gaussian} \\
\cmidrule(lr){3-5}\cmidrule(lr){6-8}\cmidrule(lr){9-11}
& & $n{=}250$ & $n{=}500$ & $n{=}750$ &
      $n{=}250$ & $n{=}500$ & $n{=}750$ &
      $n{=}250$ & $n{=}500$ & $n{=}750$ \\
\midrule
\multirow{3}{*}{\textbf{constant $\rho$}} 
& C & 100 & 100 & 100 &\;\; 0 &\;\; 0 &\;\; 0 &\;\; 5 &\;\; 0 &\;\; 0 \\
& G &\;\; 0 &\;\; 0 &\;\; 0 &  100 & 100 & 100 &\;\; 21 &\;\; 29 &\;\; 13 \\
& N &\;\; 0 &\;\; 0 &\;\; 0 &\;\; 0 &\;\; 0 &\;\; 0 &  74 &  71 & 87 \\
\midrule
\multirow{3}{*}{\textbf{varying $\rho$}} 
& C & 100 & 100 & 100 &\;\; 0 &\;\; 0 &\;\; 0 &\;\; 0 &\;\; 0 &\;\; 0 \\
& G &\;\; 0 &\;\; 0 &\;\; 0 & 94 & 100 & 99 &  0 &\;\; 0 &\;\; 0 \\
& N &\;\; 0 &\;\; 0 &\;\; 0 &\;\; 6 &\;\; 0 &\;\; 1 &  100 &  100 & 100 \\
\bottomrule
\end{tabular}

\label{tab:perc_DIC}
\end{table}

\newpage

We report the empirical evaluations for comparing $\DIC$ values under the true model and the other two competitors.\\

    \begin{figure}[h!]
    \centering
    \includegraphics[width=0.7\linewidth]{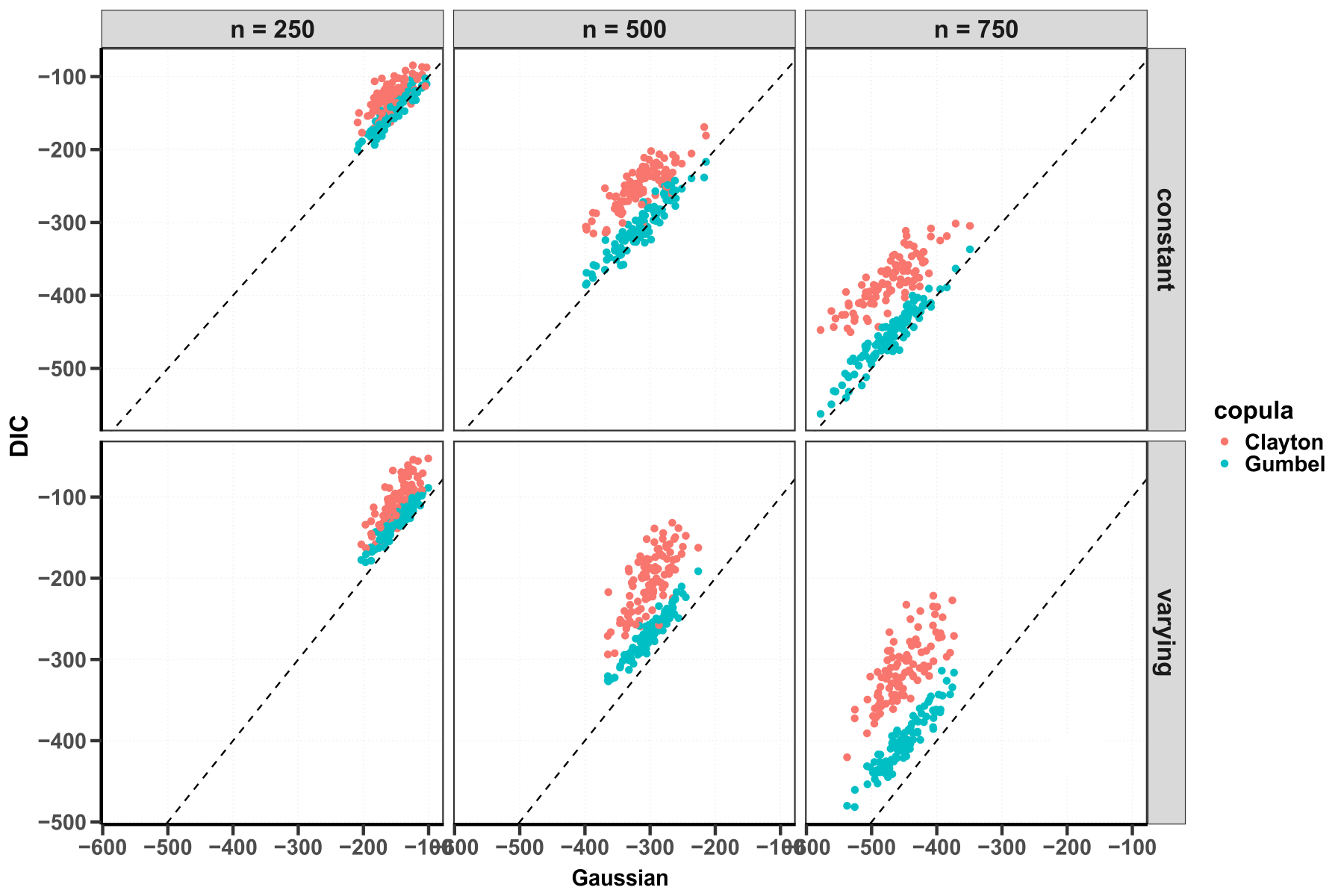}
    \caption{Comparison of $\DIC$ values under the correctly specified Gaussian copula model against the misspecified Clayton and Gumbel copula models. Columns correspond to sample sizes $n=250, 500, 750$, while rows display the two dependence scenarios: \textbf{constant} copula parameter $\rho$, and \textbf{varying} $\rho$  (covariate-dependent $\rho$).}
    \label{fig:DIC_N}
\end{figure}

       \begin{figure}[h!]
    \centering
\includegraphics[width=0.7\linewidth]{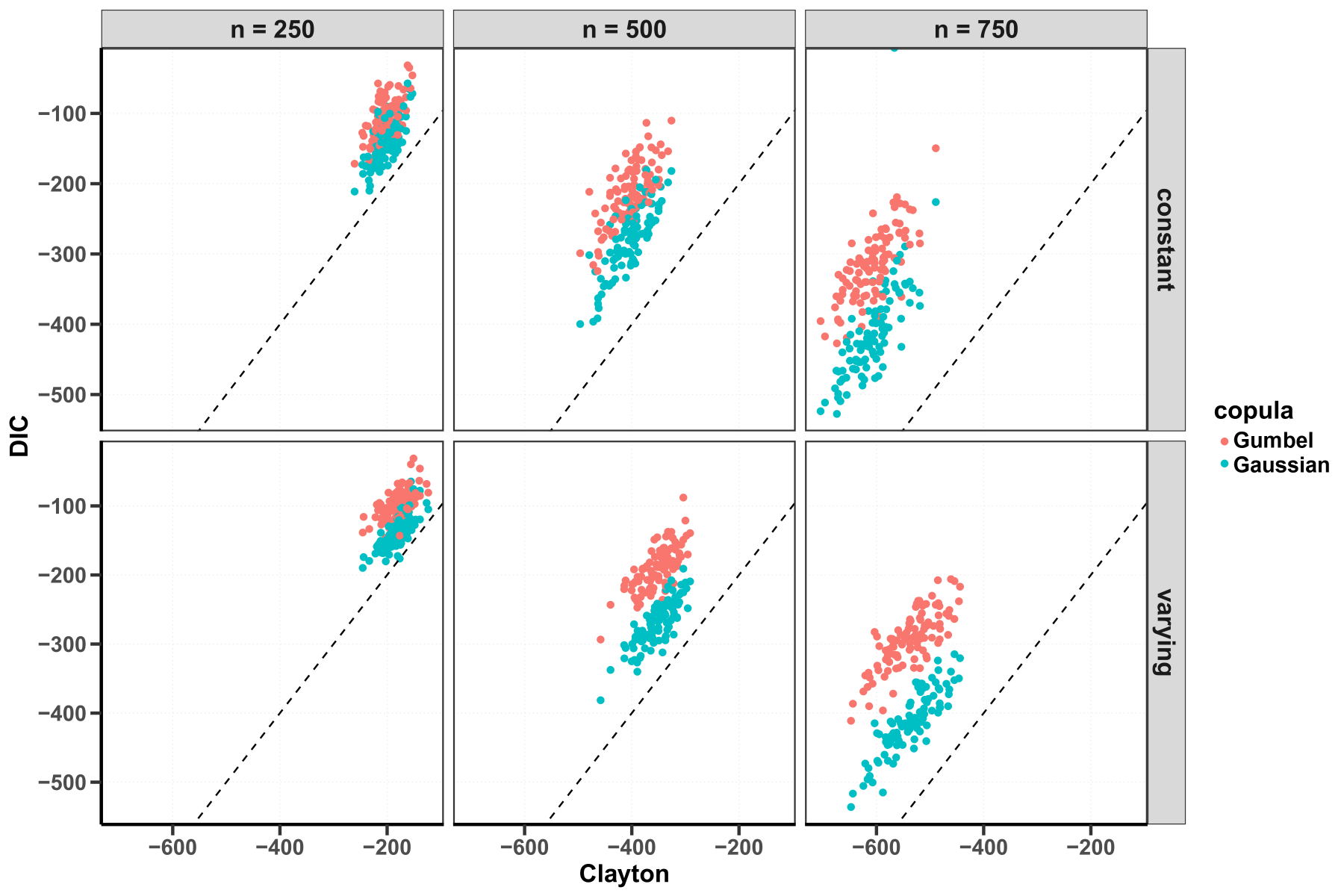}
    \caption{Comparison of $\DIC$ values under the correctly specified Clayton copula model against the misspecified Gaussian and Gumbel copula models. Columns correspond to sample sizes $n=250, 500, 750$, while rows display the two dependence scenarios: \textbf{constant} copula parameter $\rho$, and \textbf{varying} $\rho$  (covariate-dependent $\rho$).}
    \label{fig:DIC_C}
\end{figure}

\begin{figure}[h!]
    \centering \includegraphics[width=0.7\linewidth]{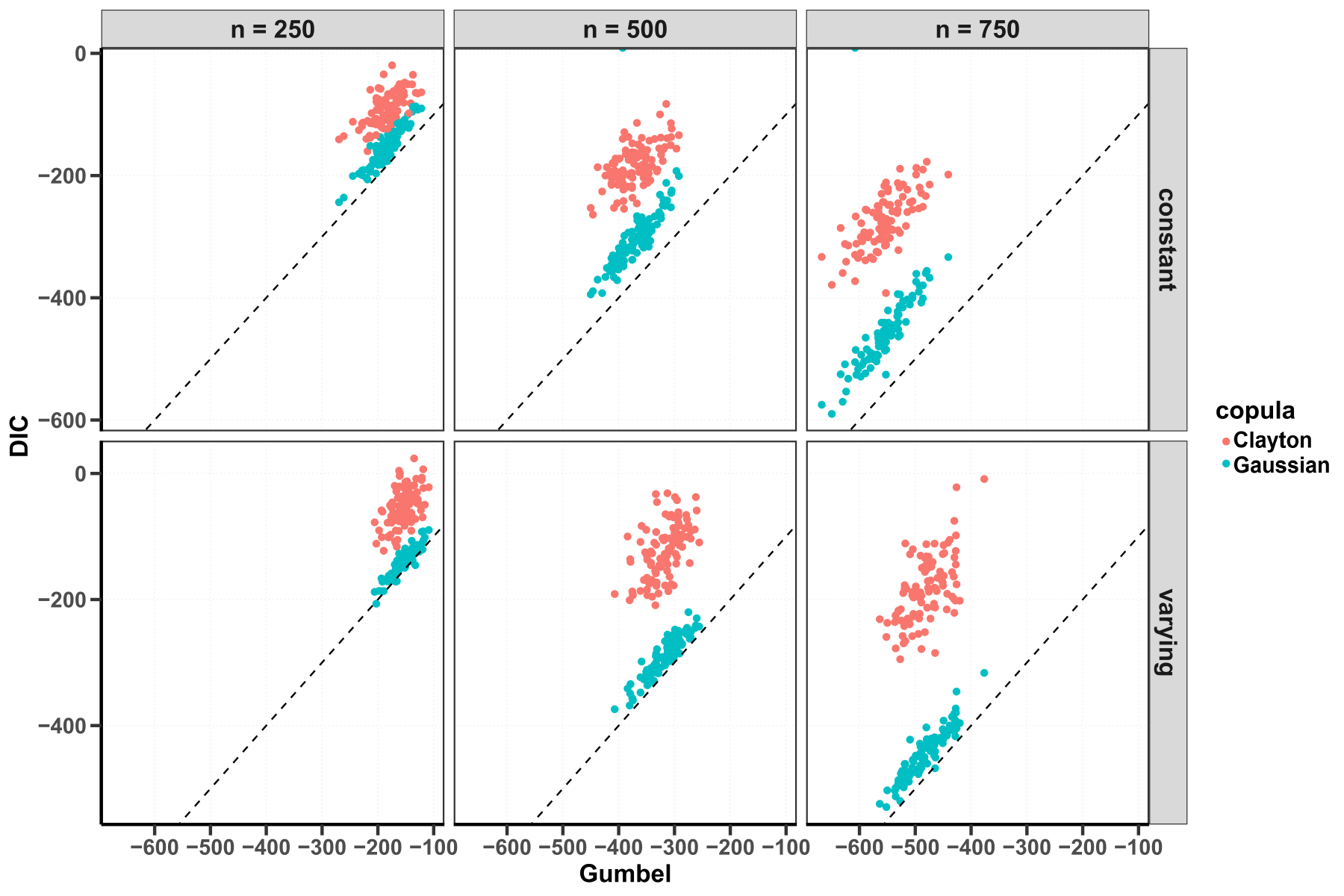}
    \caption{Comparison of $\DIC$ values under the correctly specified Gumbel copula model against the misspecified Gaussian and Clayton copula models. Columns correspond to sample sizes $n=250, 500, 750$, while rows display the two dependence scenarios: \textbf{constant} copula parameter $\rho$, and \textbf{varying} $\rho$  (covariate-dependent $\rho$).}
    \label{fig:DIC_G}
\end{figure}

\newpage 
We report the empirical evaluations for comparing $\WAIC$ values under the true model and the other two competitors.\\

\begin{table}[h!]
    \centering
    \caption{Percentage of replications with minimum $\WAIC$ among the three investigated copula~models. Columns report combinations of the true copula model under which data were generated and sample size $n$, while the rows show the fitted copula model (N="Gaussian", C="Clayton", G="Gumbel"), with a constant or covariate-dependent (varying) copula dependence parameter $\rho$.}
\begin{tabular}{l l rrrrrrrrr}
\toprule
\textbf{WAIC} & & \multicolumn{3}{c}{Clayton} & \multicolumn{3}{c}{Gumbel} & \multicolumn{3}{c}{Gaussian} \\
\cmidrule(lr){3-5}\cmidrule(lr){6-8}\cmidrule(lr){9-11}
& & $n{=}250$ & $n{=}500$ & $n{=}750$ &
      $n{=}250$ & $n{=}500$ & $n{=}750$ &
      $n{=}250$ & $n{=}500$ & $n{=}750$ \\
\midrule
\multirow{3}{*}{\textbf{constant $\rho$}} 
& C & 100 & 100 & 100 &\;\; 0 &\;\; 0 &\;\; 0 &\;\;  6 &\;\; 3 &\;\; 0 \\
& G &\;\; 0 &\;\; 0 &\;\; 0 &  100 & 100 & 100 &\;\; 47 &\;\; 51 &\;\; 51 \\
& N &\;\; 0 &\;\; 0 &\;\; 0 &\;\; 0 &\;\; 0 &\;\; 0 &   47 &  46 & 49 \\
\midrule
\multirow{3}{*}{\textbf{varying $\rho$}} 
& C & 100 & 100 & 100 &\;\; 0 &\;\; 0 &\;\; 0 &\;\; 0 &\;\; 0 &\;\; 0 \\
& G &\;\; 0 &\;\; 0 &\;\; 0 & 96 & 100 & 100 &  0 &\;\; 0 &\;\; 0 \\
& N &\;\; 0 &\;\; 0 &\;\; 0 &\;\; 4 &\;\; 0 &\;\;0 &  100 &  100 & 100 \\
\bottomrule
\end{tabular}
\label{tab:perc_WAIC}
\end{table}

\begin{figure}[h!]
    \centering
    \includegraphics[width=0.7\linewidth]{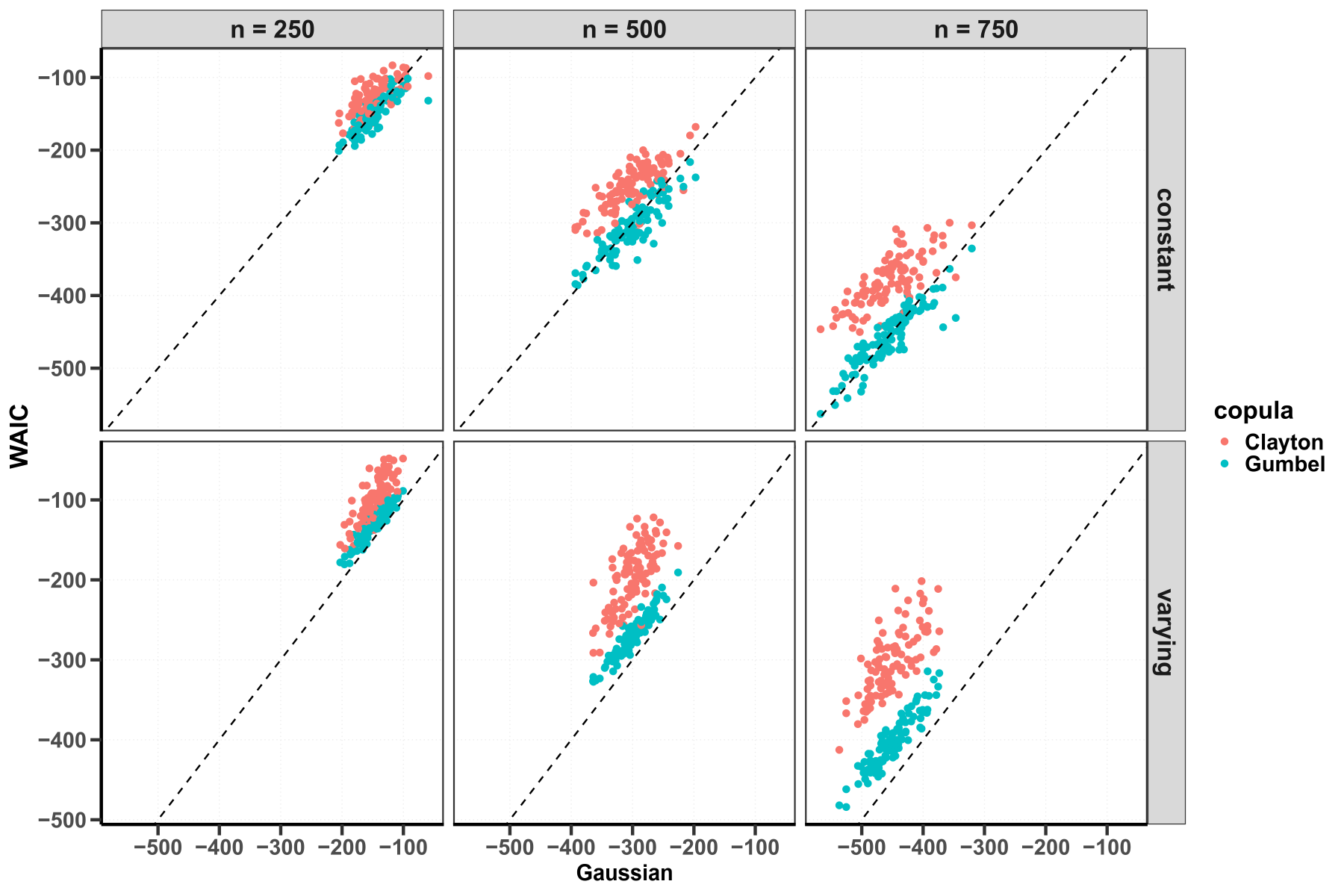}
    \caption{Comparison of $\WAIC$ values under the correctly specified Gaussian copula model (x-axis) against the misspecified Clayton and Gumbel copula models (y-axis). Columns correspond to increasing sample sizes $n=250, 500, 750$, while rows display the two dependence scenarios: \textbf{constant}, with a constant copula parameter $\rho$, and \textbf{varying}, where $\rho$ depends on a covariate.}
    \label{fig:WAIC_N}
\end{figure}

\begin{figure}[h!]
    \centering
    \includegraphics[width=0.7\linewidth]{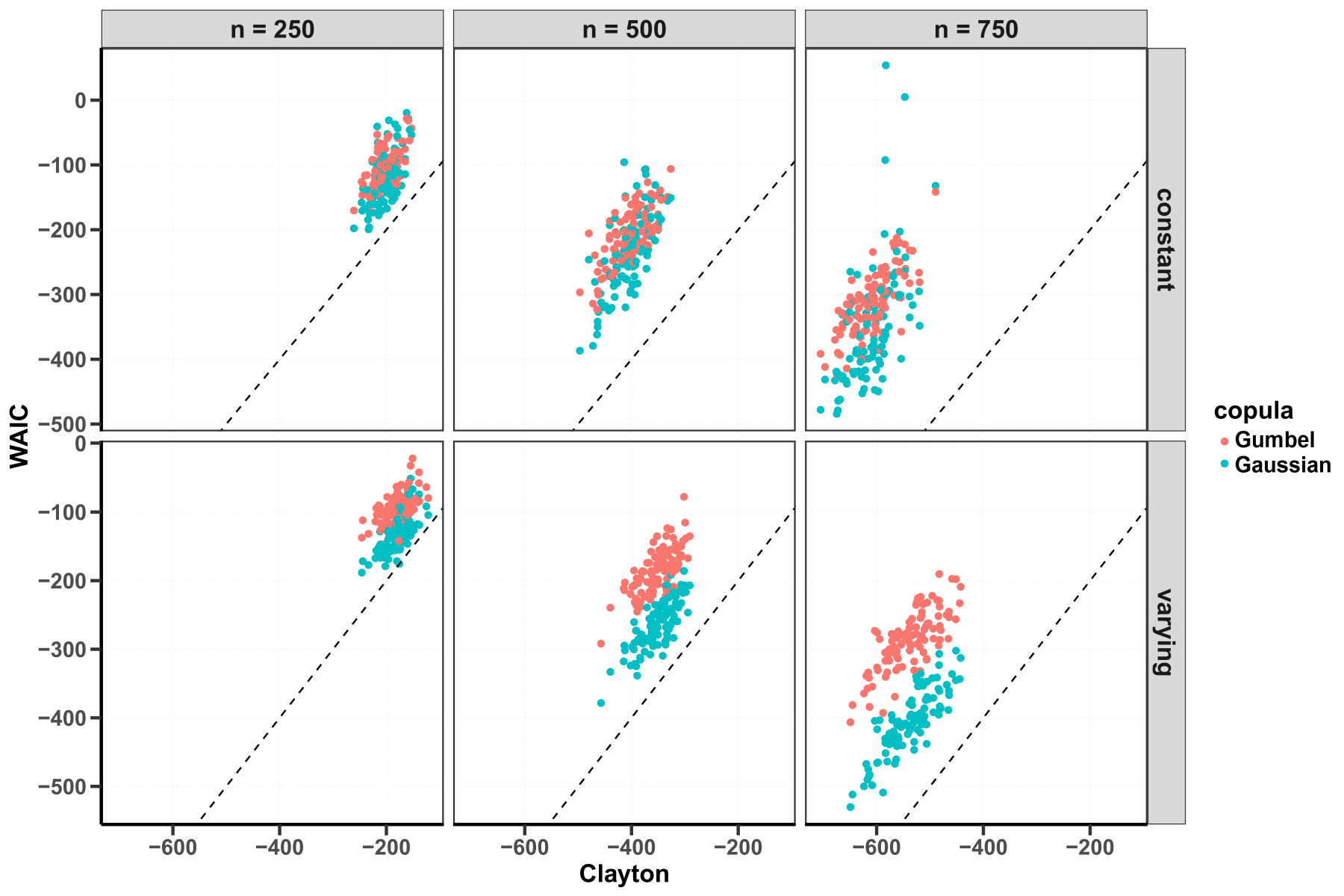}
    \caption{Comparison of $\WAIC$ values under the correctly specified Clayton copula model (x-axis) against the misspecified Gumbel and Gaussian copula models (y-axis). Columns correspond to increasing sample sizes $n=250, 500, 750$, while rows display the two dependence scenarios: \textbf{constant}, with a constant copula parameter $\rho$, and \textbf{varying}, where $\rho$ depends on a covariate.}
    \label{fig:WAIC_C}
\end{figure}

\begin{figure}[h!]
    \centering
    \includegraphics[width=0.7\linewidth]{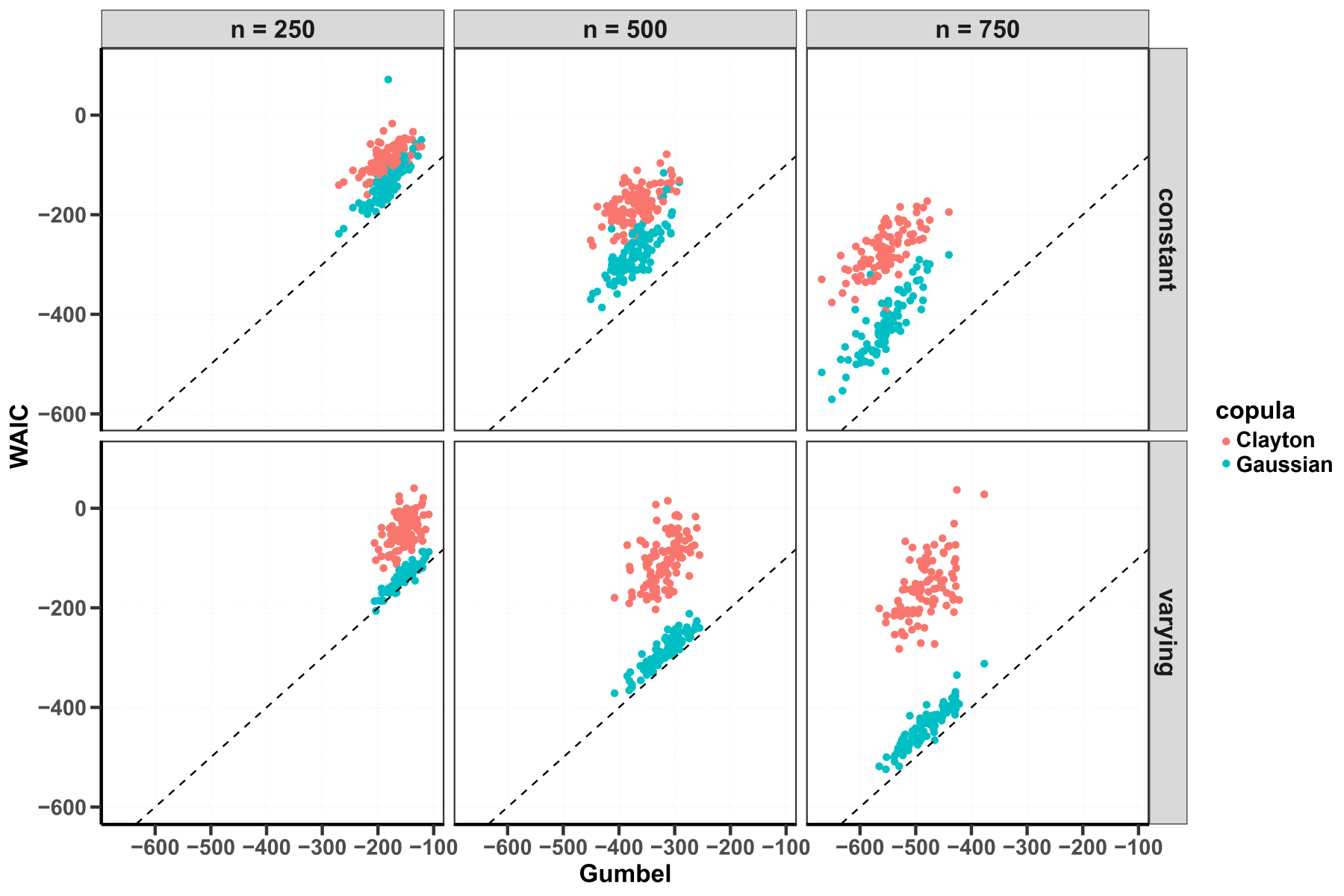}
    \caption{Comparison of $\WAIC$ values under the correctly specified Gumbel copula model (x-axis) against the misspecified Clayton and Gaussian copula models (y-axis). Columns correspond to increasing sample sizes $n=250, 500, 750$, while rows display the two dependence scenarios: \textbf{constant}, with a constant copula parameter $\rho$, and \textbf{varying}, where $\rho$ depends on a covariate.}
    \label{fig:WAIC_G}
\end{figure}

\newpage

\section{Supplementary details for Section \ref{s:appl} }\label{app:app-fig}

\begin{figure}[h!]
    \centering
        \begin{tabular}{cc}
    \textbf{ \hspace{0.25cm} \textit{   Wind direction}}  \hspace{2cm}  & \hspace{1.5cm}  \textbf{\textit{Wind speed}}\\
    \end{tabular}
    \includegraphics[width=0.85\linewidth]{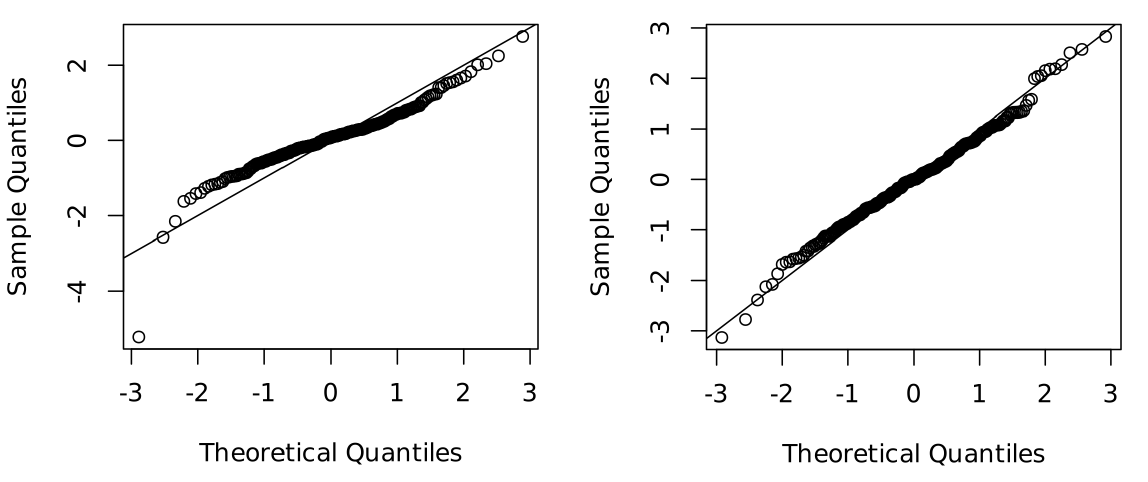}
    \caption{Randomized quantile residuals for the marginal models, corresponding to the linear model (left) for the wind speed and the circular model (right) for the wind direction.}
    \label{fig:qqplot}
\end{figure}

\end{document}